\documentclass[9pt,final,journal,twocolumn]{IEEEtran}

\IEEEoverridecommandlockouts
\usepackage{color}

\usepackage{amsthm,amsmath,amssymb}
\usepackage{mathtools}

\usepackage{verbatim}
\DeclareMathOperator{\Tr}{Tr}
\usepackage{cite}

\usepackage{url}
\usepackage{cases}
\usepackage{stmaryrd}

\usepackage[inline]{enumitem}

\usepackage{graphicx}

\usepackage{caption}
\ifCLASSOPTIONcompsoc
  \usepackage[caption=false,font=normalsize,labelfont=sf,textfont=sf]{subfig}
\else
  \usepackage[caption=false,font=footnotesize]{subfig}
\fi

\usepackage[switch]{lineno}
\usepackage[named]{algo}
\usepackage[noend]{algpseudocode}
\usepackage{algorithm}

\usepackage[flushleft]{threeparttable} 
\usepackage{bm}
\newtheorem{theorem}{Theorem}

\newtheorem{remark}{Remark}

\begin{document}
\title{Fundamental Limits for Near-Field Sensing --- Part II: Wide-Band Systems}

\author{Tong Wei,~\IEEEmembership{Member,~IEEE}, Kumar Vijay Mishra, \IEEEmembership{Senior Member,~IEEE}, 
Bhavani Shankar M.R., \IEEEmembership{Senior Member,~IEEE}, Björn Ottersten,~\IEEEmembership{Fellow,~IEEE}


\thanks{The authors are with the Interdisciplinary Centre for Security, Reliability and Trust (SnT), University of Luxembourg, Luxembourg City L-1855, Luxembourg. E-mail: \{tong.wei@, kumar-mishra@ext.,bhavani.shankar@, bjorn.ottersten@\}uni.lu.}
}

\maketitle 

\begin{abstract}
Near-field sensing with extremely large-scale antenna arrays (ELAAs) in practical 6G systems is expected to operate over broad bandwidths, where delay, Doppler, and spatial effects become tightly coupled across frequency. The purpose of this and the companion paper (Part I) is to develop the unified Cramér–Rao bounds (CRBs) for sensing systems spanning from far-field to near-field, and narrow-band to wide-band. This paper (Part II) derives fundamental estimation limits for a wide-band near-field sensing systems employing orthogonal frequency-division multiplexing signaling over a coherent processing interval. We establish an exact near-field wide-band signal model that captures frequency-dependent propagation, spherical-wave geometry, and the intrinsic coupling between target location and motion parameters across subcarriers and slow time. Similar as Part I using the Slepian–Bangs formulation, we derive the wide-band Fisher information matrix and the CRBs for joint estimation of target position, velocity, and radar cross-section, and we show how wide-band information aggregates across orthogonal subcarriers. We further develop tractable far-field and near-field approximations which provide design-level insights into the roles of bandwidth, coherent integration length, and array aperture, and clarify when wide-band effects. Simulation results validate the derived CRBs and its approximations, demonstrating close agreement with the analytical scaling laws across representative ranges, bandwidths, and array configurations.
\end{abstract}


\IEEEpeerreviewmaketitle

\section{Introduction}
Wide-band signaling is essential for high-resolution sensing and imaging because increased bandwidth improves time-delay discrimination and enables finer separation of multipath components \cite{weiss1984fundamental,yu2019wideband}. However, when combined with near-field propagation, wide-band operation introduces beam-split effect in both spatial and frequency domains \cite{cui2024near,elbir2024curse,vakalis2020fourier}. 
In particular, the spherical-wave phase and amplitude vary across both antenna elements and frequency, yielding frequency-dependent array responses and strengthening the coupling among delay, Doppler, and spatial parameters. Consequently, performance limits derived from frequency-invariant path-loss models may substantially mischaracterize achievable accuracy in wide-band near-field sensing \cite{wei2025crb}.

Near-field electromagnetic phenomena were first systematically studied in the context of antenna measurements and near-zone/far-zone characterization \cite{yaghjian1986overview}, and have since been widely explored in wireless communications with large apertures \cite{liu2023near-field} and in radar signal processing \cite{chen2002maximum}. In parallel, orthogonal frequency-division multiplexing (OFDM) is particularly relevant because it is a baseline waveform in modern cellular systems and a common choice for ISAC designs \cite{wan2024OFDM}. Its subcarrier structure naturally supports information accumulation over frequency, while its implementation imposes constraints, such as cyclic-prefix protection, subcarrier orthogonality, and the possibility of residual inter-symbol interference (ISI) or inter-carrier interference (ICI) under delay spread or mobility \cite{keskin2021limited}.

In the companion paper (Part I) \cite{wei2025PartI}, we investigated near-field sensing under narrow-band signaling and established fundamental Cramér–Rao bounds (CRBs) \cite{esmaeilbeig2022cramer,dokhanchi2019mmWave,lv2022co} that quantify how array aperture, coherent integration time, and spherical-wave curvature govern estimation accuracy as the operating regime transitions from far-field to near-field. In some practical 6G deployments, however, extremely large-scale antenna arrays (ELAAs) and synthetic-aperture operation are also expected to exploit large instantaneous bandwidths \cite{wei2023multi-IRS,wei2024RIS,hodge2023index}. In this regime, delay, Doppler, and spatial signatures are no longer cleanly separable; instead, they interact across frequency and slow time, fundamentally changing both the signal model and the resulting information structure \cite{wang2025perofrmance,wei2025PartI}. However, the fundamental limits for near-field sensing in wide-band system are still not well investigated.

Motivated by these considerations, Part II develops a CRB characterization tailored to wide-band near-field sensing with orthogonal frequency-division multiplexing (OFDM) signaling. Specifically, we establish an exact wide-band near-field observation model that captures frequency-dependent propagation and the intrinsic coupling among location, velocity, and reflectivity parameters across subcarriers and slow time. Building on the Slepian–Bangs formulation, we derive the corresponding wide-band Fisher information matrix (FIM) and CRBs, and then provide tractable far-field and near-field approximations that expose the roles of bandwidth, coherent processing interval, and aperture size.

\begin{table*}
\centering
\caption{Comparison with the state-of-the-art}
\vspace{-0.1cm}
\label{tbl:priorcomp}
\scriptsize
\begin{threeparttable}
\begin{tabular}{l|c|c|c|c|c|c|l}
\hline cf.                 & Target region & Bandwidth & Target Model & Target Number & CRB\tnote{a} ~Analysis &  Application    \\
\hline 
\hline 
\cite{wang2017coarrays}    &Far-field  & Narrow-band & Static &Single & DOA\tnote{b} ~estimation  & Radar   \\
\hline 
\cite{liang2021CramerRao}  &Far-field  & wide-band   & Static &Single & DOA estimation  & Radar  \\
\hline 
\cite{hu2024joint}         &Far-field & wide-band    & Moving &Single & Range-velocity-azimuth estimation  & ISAC\tnote{c} \\
\hline
This paper (Part I)       &Near-field & Narrow-band   & Moving &Multiple & RCS\tnote{d}, position ~and Doppler estimation  & ISAC \\
\hline
This paper (Part II)      &Near-field & wide-band   & Moving &Multiple & RCS\tnote{d}, position ~and Doppler estimation  & ISAC \\
\hline
\end{tabular}
    \begin{tablenotes}[para]
	\item[a] CRB: Cram$\acute{\mathrm e}$r-Rao bound.
	\item[b] DOA: direction of arrival.
	\item[c] ISAC: integrated sensing and communication.
	\item[d] RCS: radar cross-section.
	\end{tablenotes}
    \end{threeparttable}
    \normalsize
    \vspace{-1em}    
\end{table*}
\vspace{-0.2cm}
\subsection{Related Works}
\vspace{-0.1cm}

As mentioned in Part I, estimation-theoretic bounds, most notably the CRB, provide a principled benchmark for array-based sensing and are widely used to evaluate and guide system design. For example, in sparse-array processing, CRB analysis has been used to clarify the statistical efficiency of coarray-based DOA estimation and to reveal how virtual aperture and sensor geometry translate into fundamental accuracy limits \cite{wang2017coarrays}. Further, for broadband DOA estimation, CRBs have been derived under sub-band and frequency-decomposition models, highlighting how information aggregates across frequency and how wide-band modeling assumptions impact achievable performance \cite{liang2021CramerRao}. In OFDM-based ISAC system, CRB formulations for joint range–velocity–azimuth estimation provide a representative wide-band framework in which sensing information accumulates across subcarriers under OFDM orthogonality assumptions \cite{hu2024joint}; 
See the detailed comparison in Table~\ref{tbl:priorcomp}. 

Despite these advances, several key limitations remain when one moves to wide-band near-field sensing with synthetic-aperture and extremely large MIMO arrays: \textit{(1) Element-wise position–Doppler coupling:} Many analyses either assume static targets or adopt far-field Doppler models, overlooking the element-dependent Doppler in true near-field regimes and the resulting coupling among range, angle, and velocity; \textit{(2) Frequency-dependent propagation:} wide-band CRB studies often assume frequency-invariant path loss, which can be inaccurate in near-field wide-band operation where wavelength–distance ratios and array radiation profiles vary across subcarriers. \textit{(3) Joint multi-parameter, multi-target bounds:} Existing results frequently consider only subsets of parameters and/or a single target, rather than jointly treating position, 2D velocity, and RCS in a unified Fisher-information framework. \textit{(4) Unified regime comparison:} To the best of our knowledge, there is no systematic near-field CRB framework spanning narrow-band to wide-band and explicitly revealing the transition laws with bandwidth, aperture, and observation time. 
To fill these gaps, this two-part paper develops a unified CRB framework for near-field sensing, and the detailed contributions for Part II will be discussed in the sequel.

\vspace{-0.2cm}
\subsection{Contributions}
\vspace{-0.1cm}
In this paper (\textbf{Part II}), we establish fundamental limits for near-field sensing under wide-band OFDM signaling and provides closed-form insights suitable for system design. The contributions of this paper over prior-art are listed below:  \\
\begin{enumerate}
   \item  \textbf{Unified Near-field wide-band Sensing Model:} Prior studies on near-field sensing focus either on narrow-band or static target. In this paper (\textbf{Part II}), we propose a unified wide-band near-field ISAC model with ELAA and multiple moving targets.  Hence, the previous works \cite{liang2021CramerRao,hu2024joint} can be considered as a special case of our proposed model. also, see Table~\ref{tbl:priorcomp}.
   \item  \textbf{Unified CRB Derivation:} We derive the general expression of the Fisher Information Matrix (FIM) for multiple-target sensing, explicitly incorporating frequency-dependent propagation, spherical wavefronts, and spatial–temporal coupling across subcarriers and OFDM symbols. 
   \item  \textbf{Asymptotic Closed-form Expressions:} To gain analytical insight, we derive approximate CRB expressions under a single-target scenario for two canonical approximation regimes, i.e., wide-band far-field and wide-band near-field, demonstrating their asymptotic relationships and transitions. The asymptotic close-form expressions reveal how aperture size, signal bandwidth, and target range jointly determine the estimation accuracy and resolution limit. These results provide theoretical guidance for the design of synthetic aperture antennas which operate in near-field conditions. 
   \item  \textbf{Extensive Performance Evaluation:} Extensive numerical simulations validate the theoretical CRBs and illustrate the evolution of performance from classical far-field models to full wide-band near-field operation, serving as a benchmark for future algorithm design and measurement verification.
\end{enumerate}


\vspace{-0.2cm}
\subsection{Organization}
\vspace{-0.1cm}
The remainder of Part~II is organized as follows. Section~II introduces the system geometry and the wide-band near-field signal model. Section~III derives the Fisher information matrix, presents conditional wide-band CRBs for RCS, 2D velocity, and  location, and then provides asymptotic CRB expressions under near-field and far-field approximations and discusses the resulting scaling laws and design insights. Section~IV reports numerical experiments to validate the analysis and quantify the accuracy of the proposed approximations. Finally, Section~V concludes Part~II and briefly discusses its connections to, and differences from, Part~I.

\section{Wide-band Signal Model}
Consider a wide-band near-field sensing system employing orthogonal frequency-division multiplexing (OFDM). The transmitter (Tx) and receiver (Rx) are equipped with uniform linear arrays (ULAs) comprising $N_t$ and $N_r$ antenna elements, respectively, which may be co-located or spatially separated in a two-dimensional (2D) Cartesian plane. Let $d_t$ and $d_r$ denote the inter-element spacing of the Tx and Rx arrays. The coordinates of the $n_t$-th Tx and $n_r$-th Rx elements are denoted by ${\mathbb I}_{n_t}=(x_{n_t},y_{n_t})$ and ${\mathbb I}_{n_r}=(x_{n_r},y_{n_r})$, respectively, where $n_t\in\{1,\ldots,N_t\}$ and $n_r\in\{1,\ldots,N_r\}$. The radiative near-field coverage region, i.e., Fresnel zone, \cite{liu2023near-field} contains $Q$ sensing targets, indexed by $q\in\{1,\ldots,Q\}$ \cite{hu2024joint,hu2025high}. We then denote the location and velocity of the $q$-th sensing target by ${\mathbb I}_{q}=(x_{q},y_{q})$ and ${\bm v}^{(q)}=[v_x^{(q)},v_y^{(q)}]$, respectively \cite{sun2024widely}\footnote{In ISAC systems, scatterers are typically assumed to be static compared to the moving communication users and radar targets\cite{elbir2024terahertz,mishra2019toward}.}.

\subsection{Wide-band Transmit Signal Model}

We consider the Tx employing the OFDM waveform with $K$ subcarriers and $M$ OFDM symbols. Each OFDM symbol has a total duration $T_{\mathrm{sym}} = T_{\mathrm{cp}} + T_{\mathrm{es}}$, where $T_{\mathrm{cp}}$ and $T_{\mathrm{es}}$ denote the cyclic-prefix duration and the useful symbol duration, respectively. The subcarrier spacing is $\Delta f = 1/T_{\mathrm{es}}$. Hence, $T_{\mathrm{CPI}}=MT_{\mathrm{sym}}$ denotes the length of the coherent processing interval (CPI) \cite{liu2020super}.
Meanwhile, we denote the transmit symbol vector of the $k$-th subcarrier at $m$-th symbol slot by  ${\bf x}_{k,m}=[x_{k,m}(1),\cdots,x_{k,m}(N_t)]^T\in\mathbb{C}^{N_t\times 1}$, where $\mathbb{E}\{{\bf x}_{k,m}{\bf x}_{k,m}^H\}=\mathcal{P}_k{\bf I}_{N_t}$ denotes the orthogonal transmission and $\mathcal{P}_k$ is the average transmit power allocated to the $k$-th subcarrier \cite{keskin2021MIMO,yang2024performance}. 

Applying $N_t$ parallel $K$-point inverse fast Fourier transforms (IFFTs) to the frequency-domain symbols yields the complex baseband transmit signal for the $m$-th OFDM symbol \cite{keskin2021limited} 
\begin{align}\label{transmit_BB}
  {\bf x}_m(t)  = \sum_{k=1}^{K} {\bf x}_{k,m}e^{j2\pi{k}\triangle{f}t}\mathrm{rect}(\frac{t-mT_{\mathrm{sym}}}{T_{\mathrm{sym}}})
\end{align}
where $\mathrm{rect}(\cdot)$ is the rectangular pulse function which takes the value 1 for $t\in[0,1]$ and 0 otherwise. 
The baseband signal is then upconverted  over the block of $M$ symbols
for $t\in[0,MT_{\mathrm{sym}}]$ as \cite{zheng2017super}
\begin{align}\label{transmit_UC}
  \widetilde{\bf x}(t)  
  &=\Re\{\sum_{m=1}^{M}\sum_{k=1}^{K} {\bf x}_{k,m}e^{j2\pi{f_k}t}\mathrm{rect}(\frac{t-mT_{\mathrm{sym}}}{T_{\mathrm{sym}}})\},
\end{align}
where $f_c$ denotes the carrier frequency and $f_k=f_c+{k}\triangle{f}$ is the central frequency of the $k$-th subcarrier. 

To simplify the signal model and enable tractable analysis of the near-field wide-band sensing system, we impose the following assumptions:
\begin{description}
\item[A1] ``Bandwidth-invariant RCS'': The complex scattering coefficient of each target is assumed constant over the CPI. In particular, we assume that $\alpha_q$, $q=1,\ldots,Q$, depends only on the physical target, user, or scatterer and is invariant across all $K$ subcarriers and $M$ OFDM symbols \cite{bao2019precoding,baquero2019full}.
\item[A2] ``Constant noise power'': The noise (including thermal noise, hardware impairments, and residual self-interference) is modeled as zero-mean, spatially and temporally white, with constant power across all $K$ subcarriers and $M$ OFDM symbols \cite{cheng2021hybrid}. Furthermore, the noise is assumed uncorrelated across frequency.  
\item[A3] ``Negligible higher-order reflections'': Signal components undergoing more than one reflection are assumed to suffer severe path loss and are modeled as part of the effective noise. In particular, multi-bounce paths such as Tx–scatterer–target–Rx are neglected in the received signal model \cite{xu2023bandwidth}. 
\end{description}

\vspace{-0.2cm}
\subsection{Wide-band Near-Field Receive Signal Model}
\vspace{-0.1cm}
After the propagation delay and I/Q demodulator, the received signal at the $n_r$-th Rx antenna is
\begin{small}
\begin{align}\label{Radar_Rx}
   y_{n_r}(t) &\!=\! \sum_{q=1}^{Q}{\alpha}_q g_{n_r}^q(t)
    \sum_{n_t=1}^{N_t}g_{n_t}^q(t) \widetilde{x}_{n_t}(t\!-\!\tau_{n_tn_r}^{q})\cdot e^{j2\pi{f_{n_tn_r}^{q,k}}t}\!+\!n_{n_r}(t) 
\end{align}
\end{small}
\noindent\hspace{-1em} where ${\alpha}_q$ denotes the RCS of $q$-th target, $n_{n_r}(t)$ denotes the additive Gaussian white noise (AWGN) of the $n_r$-th Rx element in time domain, and the Doppler shift for the $q$-th target on the $k$-th subcarrier with respect to the $n_t$-th Tx and $n_r$-th Rx antenna pair is given \cite{sun2024widely}
\begin{align} \label{Doppler1}
   f_{n_tn_r}^{q,k} &= \frac{f_k}{c}(\frac{\langle{\bm v}^{(q)}, {\mathbb I}_{q}-{\mathbb I}_{n_t}\rangle}{\|{\mathbb I}_{q}-{\mathbb I}_{n_t}\|}+\frac{\langle{\bm v}^{(q)}, {\mathbb I}_{q}-{\mathbb I}_{n_r}\rangle}{\|{\mathbb I}_{q}-{\mathbb I}_{n_r}\|}), 
\end{align}
where ${\bm v}^{(q)}=[v_x^{(q)},v_y^{(q)}]$ denotes the 2D velocity vector of the $q$-th target, with $v_x^{(q)}$ and $v_y^{(q)}$ its velocity components along the $x$- and $y$-axes, respectively \cite{sun2019target}. The corresponding propagation delay in \eqref{Radar_Rx} is
\begin{align}\label{delay}
\tau_{n_tn_r}^{q}= \frac{\|{\mathbb I}_{q}-{\mathbb I}_{n_t}\|+\|{\mathbb I}_{q}-{\mathbb I}_{n_r}\|}{c},
\end{align}
and the path loss for each antenna in time domain is  
\begin{align}\label{path_loss_time}
g_{n_r}^q(t) = \sum_{k=1}^{K}g_{n_r,k}^qe^{j2\pi{f_k}t},~
g_{n_t}^q(t) = \sum_{k=1}^{K}g_{n_t,k}^qe^{j2\pi{f_k}t},
\end{align}
where $g_{n_r,k}$ and $g_{n_t,k}$ denote the path loss for $n_r$-th Rx element and $n_t$-th Tx element, respectively, on $k$-th subcarrier. 
Most existing works model the path loss in wide-band systems as frequency-independent  \cite{xiao2023novel,liu2022joint,li2021lntelligent}, i.e., $g_{n_r,1}=\cdots=g_{n_r,K}=c_0({d_0}/{\|{\mathbb I}_{q}-{\mathbb I}_{n_r}\|})^{\epsilon}$, where $c_0$ denotes the path loss at
the reference distance $d_0$ and $\epsilon$ is the path loss exponent, typically obtained from field measurements \cite{ding2016}.
However, in the near-field regime, this assumption becomes inaccurate because the wavelength-to-distance ratio varies significantly across subcarriers and antenna elements. To capture this effect, we adopt a frequency-dependent path-loss model in \eqref{Radar_Rx}, given by \cite{dovelos2021channel,xu2024near}   
\begin{flalign} \label{path_loss}
g_{n_r,k}^q = \frac{\lambda_k}{4\pi{\|{\mathbb I}_{q}-{\mathbb I}_{n_r}\|}} ,~
g_{n_t,k}^q = \frac{\lambda_k}{4\pi{\|{\mathbb I}_{q}-{\mathbb I}_{n_t}\|}}.
\end{flalign}

\begin{remark}  
Most existing works simplify the Doppler shift model based on the following approximations
\begin{align} \label{Doppler3}
   f_{n_tn_r}^{q,k} &\approx \frac{2f_k}{c}(\frac{\langle{\bm v}^{(q)}, {\mathbb I}_{q}-{\mathbb I}_{o}\rangle}{\|{\mathbb I}_{q}-{\mathbb I}_{o}\|}) \approx \frac{2f_k{v}^{(q)}}{c}, 
\end{align}
where ${v}^{(q)}=\|{\bm v}^{(q)}\|$ denotes the radial velocity of $q$-th target, in which the location and Doppler is decoupled in far-field region; See\cite{xiao2023novel,sun2019target,xu2023bandwidth}.
The first approximation in \eqref{Doppler3} assumes that all Tx/Rx elements are very close to a common reference point, while the second uses $\frac{\langle{\bm v}^{(q)}, {\mathbb I}_{q}-{\mathbb I}_{o}\rangle}{\|{\mathbb I}_{q}-{\mathbb I}_{o}\|}= \cos{\theta}\|{\bm v}^{(q)}\|$, 
where $\theta$ denotes the angle between the target velocity and the line-of-sight from the reference point to the target. Hence, the second approximation holds when the angle between the velocity vector and the line-of-sight direction is small. These assumptions are often reasonable in far-field scenarios, but they become inaccurate in the near-field region. Therefore, in this paper, we adopt the exact Doppler model in \eqref{Doppler1}, in which the target position and velocity are intrinsically coupled.
\end{remark}

\begin{figure*}[!t]
{
\begin{align}
   \tilde{y}_{n_r}(t) &=  \sum_{q=1}^{Q}{\alpha}_q g_{n_r}^q(t)
    \sum_{n_t=1}^{N_t}g_{n_t}^q(t) 
  \sum_{m=1}^{M}\sum_{k=1}^{K} {x}_{k,m}(n_t)e^{j2\pi{f_k}(t-\tau_{n_tn_r}^{q})}e^{j2\pi{f_{n_tn_r}^{q,k}}t}\mathrm{rect}(\frac{t-mT_{\mathrm{sym}}}{T_{\mathrm{sym}}})+n_{n_r}(t), \label{Radar_Rx_R1}\\
  y_{n_r}(t)  &=  \sum_{q=1}^{Q}{\alpha}_q g_{n_r}^q(t)
    \sum_{n_t=1}^{N_t}g_{n_t}^q(t) 
  \sum_{m=1}^{M}\sum_{k=1}^{K} {x}_{k,m}(n_t)e^{j2\pi{f_k}(t-\tau_{n_tn_r}^{q})}e^{j2\pi{f_{n_tn_r}^{q,k}}mT_{\mathrm{sym}}}\mathrm{rect}(\frac{t-mT_{\mathrm{sym}}}{T_{\mathrm{sym}}})+n_{n_r}(t), \label{Radar_Rx_R2} \\
     y_{n_r}[m,k] 
   & = \int_{mT_{\mathrm{sym}}+T_{\mathrm{cp}}}^{(m+1)T_{\mathrm{sym}}}y_{n_r}(t)e^{-2j\pi f_kt}dt =\int_{mT_{\mathrm{sym}}+T_{\mathrm{cp}}}^{(m+1)T_{\mathrm{sym}}}y_{n_r}(t)e^{-2j\pi f_ct}e^{-2j\pi k\triangle ft}dt \nonumber\\
 &= \sum_{q=1}^{Q}{\alpha}_q g_{n_r,k}^q e^{-j2\pi{f_k}(\frac{\|{\mathbb I}_{q}-{\mathbb I}_{n_r}\|}{c})}e^{j2\pi\frac{f_k}{c}(\frac{\langle{\bm v}^{(q)}, {\mathbb I}_{q}-{\mathbb I}_{n_r}\rangle}{\|{\mathbb I}_{q}-{\mathbb I}_{n_r}\|})mT_{\mathrm{sym}}} {\bf a}_{\mathrm{Tx}}^T(f_k,{\mathbb I}_{q},{\bm v}^{(q)}){\bf x}_{k,m}+n_{n_r}[m,k], \label{Radar_Rx_R3}
\end{align}}
\hrulefill
\end{figure*}

Substituting \eqref{transmit_UC} into \eqref{Radar_Rx}, the receive signal of $n_r$-th Rx antenna in time domain is given by \eqref{Radar_Rx_R1} which is arranged at the top of next page.
Following \cite{zheng2017super,keskin2021limited}, we then assume that the Doppler-related phase-shift within an OFDM symbol duration is negligible, i.e., $f_{n_tn_r}^{q,k}\ll\frac{1}{T_{\mathrm{sym}}}<\frac{1}{T_{\mathrm{es}}}=\triangle f$, hence $f_{n_tn_r}^{q,k}t\approx f_{n_tn_r}^{q,k}mT_{\mathrm{sym}}$, $t\in[mT_{\mathrm{sym}},(m+1)T_{\mathrm{sym}}]$\footnote{This is because that for the function $f_{n_tn_r}^{q,k}(t)$, the term $f_{n_tn_r}^{q,k}=\tan{\psi}$ denotes the slope of the function. If time difference is small enough, the maximum value difference is also negligible,  $f_{n_tn_r}^{q,k}T_{\mathrm{sym}}\ll1$ in each OFDM symbol duration.}. Then, we can rewrite \eqref{Radar_Rx_R1} into \eqref{Radar_Rx_R2} which is also arranged at the top of next page. 
Removing the carrier frequency and CP and implementing the Fourier transform within $m$-th time interval i.e., OFDM symbol duration, the received signal of $m$-th OFDM symbol on the $k$-th subcarrier is given by \eqref{Radar_Rx_R3}, see the top of next page \cite{keskin2021limited},
where $n_{n_r}[m,k]=\int_{mT_{\mathrm{sym}}+T_{\mathrm{cp}}}^{(m+1)T_{\mathrm{sym}}}n_{n_r}(t)e^{-2\pi f_ct}e^{-2\pi k\triangle ft}dt$ denotes the radar noise at the $m$-th ODFM symbol on the $k$-th subcarrier which follows zero mean and variance matrix $\sigma_r^2$. 
Stacking the received signals of all Rx elements, we can rewrite \eqref{Radar_Rx_R3} as
\begin{align}\label{Radar_Rx_R4}
  {\bf y}_{k,m} \!=\! \sum_{q=1}^{Q} {\alpha}_q {\bf a}_{\mathrm{Rx}}(k,m,q){\bf a}_{\mathrm{Tx}}^T(k,m,q){\bf x}_{k,m}\!+\!{\bf n}_{k,m},
\end{align}
where
\begin{flalign*}
  {\bf a}_{\mathrm{Tx}}(k,m,q) =& [g_{1,k}^qe^{j2\pi\varphi_1(k,m,q)},\cdots,g_{N_t,k}^qe^{j2\pi\varphi_{N_t}(k,m,q)}]^T, \\  {\bf a}_{\mathrm{Rx}}(k,m,q)=& [g_{1,k}^qe^{j2\pi\varphi_1(k,m,q)},\cdots,g_{N_r,k}^qe^{j2\pi\varphi_{N_r}(k,m,q)}]^T,
\end{flalign*}
denote the wide-band near-field steering vector of Tx and Rx, respectively, in which 
\begin{align}
&\varphi_{n_r}(k,m,q)\!=\!{\frac{f_k}{c}((\frac{\langle{\bm v}^{(q)}, {\mathbb I}_{q}\!-\!{\mathbb I}_{n_r}\rangle}{\|{\mathbb I}_{q}-{\mathbb I}_{n_r}\|})mT_{\mathrm{sym}}\!-\!\|{\mathbb I}_{q}\!-\!{\mathbb I}_{n_r}\|)}  \nonumber
\end{align}
denotes the total phase shift caused by the location and Doppler of $q$-th target for $n_r$-th elements at $m$-th OFDM symbol and $k$-th subcarrier. It is observed from \eqref{Radar_Rx_R4} that in the near-field region, Doppler shift is related to both the OFDM symbol and subcarrier index. Note that the far-field steering vector can be considered as a special case  when the distance between antenna and target is sufficiently large \cite{xu2024near}. Thus, the proposed method can be also utilized in the far-field sensing system. 

To further simplify the notation, the receive signal in \eqref{Radar_Rx_R4} is rewritten  as 
\begin{flalign}\label{Radar_Rx_R6}
  {\bf y}_{k,m} = & {\bf A}_{\mathrm{Rx},k,m}{\bf B} {\bf A}_{\mathrm{Tx},k,m}^T {\bf x}_{k,m}\!+\!{\bf n}_{k,m} \nonumber\\
  = &{\bf A}_{k,m}{\bf x}_{k,m}\!+\!{\bf n}_{k,m},
\end{flalign}
where ${\bf B}=\mathrm{diag}[\alpha_1,\cdots,\alpha_{Q}]$, ${\bf A}_{\mathrm{Tx},k,m}=[{\bf a}_{\mathrm{Tx}}(k,m,1),\cdots,{\bf a}_{\mathrm{Tx}}(k,m,Q)]$ and ${\bf A}_{\mathrm{Rx},k,m}=[{\bf a}_{\mathrm{Rx}}(k,m,1),\cdots,{\bf a}_{\mathrm{Rx}}(k,m,Q)]$.
We can further stack the received signal in \eqref{Radar_Rx_R6} from all $M$ OFDM symbols into a single vector as 
\begin{align}\label{Radar_Rx_R7}
  {\bf y}_{k} = {\bf A}_{k}
     {\bf x}_{k}+{\bf n}_{k}
\end{align}
where ${\bf y}_{k}=[{\bf y}_{k,1}^T,\cdots,{\bf y}_{k,M}^T]^T$ denotes the receive signal of $k$-th subcarrier, ${\bf A}_{k}=\mathrm{bdiag}[{\bf A}_{k,1},\cdots,{\bf A}_{k,M}]$, ${\bf x}_{k}=[{\bf x}_{k,1}^T,\cdots,{\bf x}_{k,M}^T]^T$ and ${\bf n}_{k}=[{\bf n}_{k,1}^T,\cdots,{\bf n}_{k,M}^T]^T$. 
Then, we collect all unknown real-valued parameters in
\begin{align}\label{parameter_real}
{\bm\theta} &=[{\bm x},{\bm y},{\bm v}_x,{\bm v}_y,\bm\alpha_{R},\bm\alpha_{I}]^T\in \mathbb{R}^{6Q\times1},
\end{align}
where ${\bm x}=[x_1,\cdots, x_{Q}]$, ${\bm y}=[y_1,\cdots, y_{Q}]$, ${\bm v}_x=[v_x^{(1)},\cdots,v_x^{(Q)}]$, ${\bm v}_y=[v_y^{(1)},\cdots,v_y^{(Q)}]$, $\bm\alpha_{R}=[\alpha_{R_1},\cdots,\alpha_{R_Q}]$, and $\bm\alpha_{I}=[\alpha_{I_1},\cdots,\alpha_{I_Q}]$.  Thus, $\bm\theta$ jointly parameterizes the $x$- and $y$-coordinates, the $x$- and $y$-components of velocity, and the real and imaginary parts of the complex RCS coefficients for all $Q$ targets.

\begin{remark}[Wide-band implication and ISI condition]
In the wide-band OFDM case, targets imprint frequency-dependent delay signatures across subcarriers through the factor $e^{-j2\pi f_k \tau}$, so bandwidth provides an additional information dimension beyond spatial aperture and slow-time Doppler. This advantage relies on the OFDM orthogonality assumption that the cyclic prefix (CP) exceeds the effective delay spread. Specifically, to avoid inter-symbol interference (ISI), it is typically required that
\begin{align}
\max_{q,n_t,n_r}\ \tau_{n_tn_r}^{(q)} - \min_{q,n_t,n_r}\ \tau_{n_tn_r}^{(q)} \ \le\ T_{\mathrm{cp}},
\end{align}
(or, more generally, that the overall channel delay spread observed at the receiver is within $T_{\mathrm{cp}}$). If this condition is violated---which can occur in near-field wide-aperture deployments or rich multipath environments---residual ISI may arise and a time-domain model/equalization (or a longer CP) becomes necessary. In this paper, we assume the CP is sufficiently long such that ISI is negligible, enabling the per-subcarrier discrete model used for CRB analysis.
\end{remark}

\section{Wide-band CRB Analysis}
To evaluate the parameter estimation performance, CRB has been widely utilized as a estimation lower-bound \cite{stoica1990performance,liang2021CramerRao}. 
In this section, we derive the wide-band near-field CRBs for joint estimation of position, velocity, and RCS of moving targets. 
Let us first rewrite \eqref{Radar_Rx_R7} as the real vector $\bar{\bf y}_{k}=[\Re{({\bf y}_{k})}^T;\Im{({\bf y}_{k})^T}]^T$, $\bar{\bm\mu}_k=[\Re{({\bm\mu}_k)}^T,\Im{({\bm\mu}_k)}^T]^T$, and $\bar{\bf C}_k=\frac{\sigma_r^2}{2}{\bf I}_{2N_rM}$,
where ${\bm\mu}_k$ and ${\bf C}_k$ denote the mean and covariance of ${\bf y}_k$, respectively, i.e., ${\bf y}_k\sim({\bm\mu}_k,{\bf C}_k)$, and with the AWGN assumption\footnote{Note that \eqref{dist._rx} holds due to the AWGN assumption. If this assumption does not hold, i.e., not circular or white, the additional parameters are required to characterize the matrix ${\bf C}_k$, which beyond the scope of this paper and will leave for future works.}, we have 
\begin{subequations} \label{dist._rx}
\begin{align}
 {\bm\mu}_k &= {\bf A}_k{\bf x}_k, \label{statistic1}\\
 {\bf C}_k  &= \mathbb{E}\{({\bf y}_k-{\bm\mu}_k)({\bf y}_k-{\bm\mu}_k)^H\}=\sigma_r^2{\bf I}_{N_r M}, \label{statistic2}
\end{align}
\end{subequations}
To evaluate CRB, we can write the likelihood function as \cite{stoica2005spectral}
\begin{align}\label{likelihood}
 p_k(\bar{\bf y}_k,{\bm\theta}) = \frac{1}
 {\sqrt{2\pi^{2{N}_rM} |\bar{\bf C}_k|}}
 e^{-(\bar{\bf y}_k-\bar{\bm\mu}_k)^T\bar{\bf C}_k^{-1}(\bar{\bf y}_k-\bar{\bm\mu}_k)/2}. 
\end{align}
According to \eqref{likelihood}, we can write the log-likelihood function  as 
\begin{flalign}\label{likelihood_log}
\ln{p_k} 
 \!=\! -\frac{1}{2}(\ln{|\bar{\bf C}_k|}\!+\!
 (\bar{\bf y}_k\!-\!\bar{\bm\mu}_k)^T\bar{\bf C}_k^{-1}(\bar{\bf y}_k\!-\!\bar{\bm\mu}_k)), 
\end{flalign}
in which the constant term $MN_r\ln{2\pi}$ is omitted for simplicity hereafter. 
Based on the above, the near-field FIM at $k$-th subcarrier is given by \cite{stoica2005spectral}
\begin{align}\label{FIM_gen}
 {\bf F}_k = -\mathbb{E}\left\{ \frac{\partial^2\ln{p_k}}{\partial{\bm\theta}~\partial{\bm\theta}^T}\right\} 
\end{align}
where $\mathbb{E}\{\cdot\}$ denotes the expectation operator. Because OFDM orthogonalizes the subcarriers and we assume no inter-carrier interference, the observations at different subcarriers are independent \cite{zhang2020joint}. Thus, the wide-band near-field FIM is therefore the sum over subcarriers \cite{liang2021CramerRao}
\begin{align}\label{FIM_wide-band}  
   \widetilde{\bf F} = \sum_{k=1}^{K}{\bf F}_k.   
\end{align} 
Based on Slepian-Bangs formula \cite{esmaeilbeig2022cramer}, we can calculate the inverse of generalized CRB on $k$-th subcarrier, see equation (B.3.25) in \cite{stoica2005spectral}, as  
\begin{small}
\begin{flalign}\label{FIM_spe}
&{\bf F}_k(i,j) \!=\!\Tr\left({\bf C}_k^{-1}\frac{\partial{\bf C}_k}{\partial{\theta}_i}{\bf C}_k^{-1}\frac{\partial{\bf C}_k}{\partial{\theta}_j}\right)\!+\!2\Re\left[\left({\frac{\partial{\bm\mu}_k}{\partial{\theta}_i}}\right)^H{\bf C}_k^{-1}\frac{\partial{\bm\mu}_k}{\partial{\theta}_j}\right],& 
\end{flalign}
\end{small}
in which the first term is equal to zero due to $\frac{\partial{\bf C}_k}{\partial{\theta}_k}={\bf 0}, \forall~k$.   
Then we can substitute the exact expression of ${\bf C}_k$ and ${\bm\mu}_k$ into \eqref{FIM_spe}, we have 
\begin{align} \label{derivation3}
   \frac{\partial{\bm\mu}_k}{\partial{\theta}_i} =   
    \left[\begin{matrix}  
       \frac{\partial{\bf A}_{k,1}{\bf x}_{k,1}}{\partial{\theta}_i}    \\
       \vdots   \\
      \frac{\partial{\bf A}_{k,M}{\bf x}_{k,M}}{\partial{\theta}_i} 
    \end{matrix}\right],  ~
   \frac{\partial{\bm\mu}_k}{\partial{\theta}_j} = 
    \left[\begin{matrix}  
       \frac{\partial{\bf A}_{k,1}{\bf x}_{k,1}}{\partial{\theta}_j}    \\
       \vdots   \\
      \frac{\partial{\bf A}_{k,M}{\bf x}_{k,M}}{\partial{\theta}_j} 
    \end{matrix}\right], 
\end{align}
where ${\theta}_i$ and ${\theta}_j$ denotes the $i$-th and $j$-th
element of ${\bm\theta}$.
Notice that once the equation \eqref{statistic2} holds, the first term in \eqref{FIM_spe} will be vanished and ${\bf C}_k^{-1}=\frac{1}{\sigma_r}{\bf I}$. Then, the difficulty of FIM calculation is how to calculate the derivation vectors \eqref{derivation3}. Now, we will present how to calculate the vectors in \eqref{derivation3} in the sequel. 
The full FIM at the $k$-th subcarrier can be partitioned according to the structure of $\bm\theta$ as
\begin{small}
\begin{align}\label{FIM_block}
\hspace*{-0.5em}{\bf F}_k  
\!=\!&\left[\begin{matrix}  
 {\bf F}^k_{{\bf x}{\bf x}} \!&\! {\bf F}^k_{{\bf x}{\bf y}} \!&\! {\bf F}^k_{{\bf x}{\bf v}_x} \!&\! {\bf F}^k_{{\bf x}{\bf v}_y} \!&\! {\bf F}^k_{{\bf x}{\bm\alpha}_R} \!&\! {\bf F}^k_{{\bf x}{\bm\alpha}_I}  \\
 {\bf F}^k_{{\bf y}{\bf x}} \!&\! {\bf F}^k_{{\bf y}{\bf y}} \!&\! {\bf F}^k_{{\bf y}{\bf v}_x} \!&\! {\bf F}^k_{{\bf y}{\bf v}_y} \!&\! {\bf F}^k_{{\bf y}{\bm\alpha}_R} \!&\! {\bf F}^k_{{\bf y}{\bm\alpha}_I}\\
  {\bf F}^k_{{\bf v}_x{\bf x}} \!&\! {\bf F}^k_{{\bf v}_x{\bf y}} \!&\! {\bf F}^k_{{\bf v}_x{\bf v}_x} \!&\! {\bf F}^k_{{\bf v}_x{\bf v}_y} \!&\! {\bf F}^k_{{\bf v}_x{\bm\alpha}_R} \!&\! {\bf F}^k_{{\bf v}_x{\bm\alpha}_I}\\
  {\bf F}^k_{{\bf v}_y{\bf x}} \!&\! {\bf F}^k_{{\bf v}_y{\bf y}} \!&\! {\bf F}^k_{{\bf v}_y{\bf v}_x} \!&\! {\bf F}^k_{{\bf v}_y{\bf v}_y} \!&\! {\bf F}^k_{{\bf v}_y{\bm\alpha}_R} \!&\! {\bf F}^k_{{\bf v}_y{\bm\alpha}_I}\\
  {\bf F}^k_{{\bm\alpha}_R{\bf x}} \!&\! {\bf F}^k_{{\bm\alpha}_R{\bf y}} \!&\! {\bf F}^k_{{\bm\alpha}_R{\bf v}_x} \!&\! {\bf F}^k_{{\bm\alpha}_R{\bf v}_y} \!&\! {\bf F}^k_{{\bm\alpha}_R{\bm\alpha}_R} \!&\! {\bf F}^k_{{\bm\alpha}_R{\bm\alpha}_I}\\
  {\bf F}^k_{{\bm\alpha}_I{\bf x}} \!&\! {\bf F}^k_{{\bm\alpha}_I{\bf y}} \!&\! {\bf F}^k_{{\bm\alpha}_I{\bf v}_x} \!&\! {\bf F}^k_{{\bm\alpha}_I{\bf v}_y} \!&\! {\bf F}^k_{{\bm\alpha}_I{\bm\alpha}_R} \!&\! {\bf F}^k_{{\bm\alpha}_I{\bm\alpha}_I}\\
    \end{matrix}\right],
\end{align}
\end{small}
\noindent\hspace{-1em} where the diagonal blocks correspond to the information on position, velocity, and RCS parameters, and the off-diagonal blocks capture their mutual coupling. For an unbiased estimator, the CRB matrix is given by the inverse FIM, i.e., ${\mathrm {CRB}}=\widetilde{\bf F}^{-1}$, where $\widetilde{\bf F}$ is defined in \eqref{FIM_wide-band}. When focusing on a subset of parameters (e.g., position only), we use Schur-complement techniques \cite{dokhanchi2019mmWave} to extract the corresponding sub-block of $\widetilde{\bf F}^{-1}$, which yields the tightest CRB for that subset. 


\vspace{-0.2cm}
\subsection{Conditional CRB for RCS Estimation}
\vspace{-0.1cm}

To begin with, we have the following equations for the complex-valued partial derivative in terms a real value \cite{hjorungnes2007complex}
\begin{flalign} \label{FIM_Derivation}
  &\frac{\partial(a+jb)e^{jc}}{\partial a} = e^{jc},~
  \frac{\partial(a+jb)e^{jc}}{\partial b} = je^{jc}. 
\end{flalign}
Then, we can calculate the partial derivative related to the real and imaginary parts of RCS as 
\begin{subequations} 
\label{FIM_Derivation_RCS}
\begin{align} 
 & \frac{\partial{\bf A}_{k,m}{\bf x}_{k,m}}{\partial\alpha_{R_{q}}} =  {\bf a}_{\mathrm{R}}(k,m,q){\bf a}^T_{\mathrm{T}}(k,m,q){\bf x}_{k,m}, \label{FIM_Derivation_RCS1}\\
& \frac{\partial{\bf A}_{k,m}{\bf x}_{k,m}}{\partial\alpha_{I_{q}}} = j {\bf a}_{\mathrm{R}}(k,m,q){\bf a}^T_{\mathrm{T}}(k,m,q){\bf x}_{k,m}.\label{FIM_Derivation_RCS2}
\end{align}
\end{subequations}
where $q=1,\cdots,Q$.
According to Cauchy-Riemann equations, we have  \cite{hjrungnes2011complex}
\begin{subequations} \label{FIM_RCS}
\begin{align} 
  \frac{\partial{\bf A}_{k,m}{\bf x}_{k,m}}{\partial{\bm\alpha}_{R}} &= [\breve{\bf A}_{k,m}(1){\bf x}_{k,m},\cdots,\breve{\bf A}_{k,m}(Q){\bf x}_{k,m}],  \label{FIM_RCS_Re}\\
  \frac{\partial{\bf A}_{k,m}{\bf x}_{k,m}}{\partial{\bm\alpha}_{I}} &= j[\breve{\bf A}_{k,m}(1){\bf x}_{k,m},\cdots,\breve{\bf A}_{k,m}(Q){\bf x}_{k,m}],   \label{FIM_RCS_Im}
\end{align}
\end{subequations}
where $\breve{\bf A}_{k,m}(q)={\bf A}_{\mathrm{R},k,m}{\bf \Lambda}_{q} {\bf A}_{\mathrm{T},k,m}^T$ denotes the partial derivative matrix and ${\bf \Lambda}_{q}$ denotes the selection matrix with $q$-th diagonal element is one and the others are zero. Substituting \eqref{FIM_RCS} into \eqref{FIM_wide-band} and \eqref{FIM_spe}, we can get the wide-band CRB for RCS estimation as \cite{liang2021CramerRao} 
\begin{small}
\begin{align} \label{CRB_resp.2.2}
&\mathrm{CRB}_{\alpha_{R_{q}}} \!=\!\mathrm{CRB}_{\alpha_{I_{q}}} \nonumber \\
= &\frac{\sigma_r^2}{2\mathbb{E}\{{\sum}_{k=1}^K{\sum}_{m=1}^M\|{\bf a}_{\mathrm{R}}(k,m,q){\bf a}_{\mathrm{T}}^{T}(k,m,q){\bf x}_{k,m}\|_2^2\}}. 
\end{align}
\end{small}
Then, we have $\mathrm{CRB}_{\alpha_{R_{q}}}=\mathrm{CRB}_{\alpha_q}+\mathrm{CRB}_{\alpha_{I_{q}}}$.


\vspace{-0.2cm}
\subsection{Conditional CRB for Velocity Estimation}
\vspace{-0.1cm}
Similarly, the velocity part is 
\begin{subequations} \label{Derivation_vel.}
\begin{align} 
  \frac{\partial{\bf A}_{k,m}{\bf x}_{k,m}}{\partial{\bf v}_x} &
  = [\acute{\bf A}_{k,m}(1){\bf x}_{k,m},\cdots,\acute{\bf A}_{k,m}(Q){\bf x}_{k,m}], \label{Derivation_vel.1} \\
  \frac{\partial{\bf A}_{k,m}{\bf x}_{k,m}}{\partial{\bf v}_y} &
  =[\grave{\bf A}_{k,m}(1){\bf x}_{k,m},\cdots,\grave{\bf A}_{k,m}(Q){\bf x}_{k,m}], \label{Derivation_vel.2}
\end{align}
\end{subequations}
where $\acute{\bf A}_{k,m}(q)=[{\acute{\bf A}_{\mathrm{R},k,m}}{\bf \Lambda}_{q} {\bf B}{\bf A}_{\mathrm{T},k,m}^T+{\bf A}_{\mathrm{R},k,m}{\bf B}{\bf \Lambda}_{q}{\acute{\bf A}_{\mathrm{T},k,m}^T}]$ and $\grave{\bf A}_{k,m}(q)$ denotes the partial derivative of ${\bf A}_{k,m}$ w.r.t. $v_{x_q}$ and $v_{y_q}$, respectively, in which any of these is given  
\begin{align*} 
\frac{\partial{\bf a}_{\mathrm{R},\mathrm{T},\ast}(k,m,q) }{\partial v_{x_q}}&= \frac{j2\pi\frac{f_k}{c}mT_{\mathrm{sym}}({x}_q-{x}_{\ast})}{{\|{\mathbb I}_{q}-{\mathbb I}_{\ast}\|}}{\bf a}_{\mathrm{R},\mathrm{T},\ast}(k,m,q)\\
\frac{\partial{\bf a}_{\mathrm{R},\mathrm{T},\ast}(k,m,q) }{\partial v_{y_q}}&= \frac{j2\pi\frac{f_k}{c}mT_{\mathrm{sym}}({y}_q-{y}_{\ast})}{{\|{\mathbb I}_{q}-{\mathbb I}_{\ast}\|}}{\bf a}_{\mathrm{R},\mathrm{T},\ast}(k,m,q).
\end{align*}
Substituting \eqref{Derivation_vel.} into \eqref{FIM_spe} and \eqref{FIM_wide-band}, the CRB related to the velocity estimation of $q$-th target is 
\begin{subequations} \label{CRB_vel.q}
\begin{align}
 \mathrm{CRB}_{v_{x_q}}&= \frac{\sigma_r^2}{2\mathbb{E}\{\sum_{k=1}^K\sum_{m=1}^M\|\acute{\bf A}_{k,m}(q){\bf x}_{k,m}\|_2^2\}} ,    \label{CRB_vel.q1}\\
 \mathrm{CRB}_{v_{y_q}}& = \frac{\sigma_r^2}{2\mathbb{E}\{\sum_{k=1}^K\sum_{m=1}^M\|\grave{\bf A}_{k,m}(q){\bf x}_{k,m}\|_2^2\}}. \label{CRB_vel.q2}
\end{align}
\end{subequations}

\vspace{-0.2cm}
\subsection{Conditional CRB for Location Estimation}
\vspace{-0.1cm}

Similarly, for the partial derivative w.r.t location, we have 
\begin{subequations} \label{Derivation_loc.}
\begin{align} 
  \frac{\partial{\bf A}_{k,m}{\bf x}_{k,m}}{\partial{\bf x}} &=[\ddot{\bf A}_{k,m}(1){\bf x}_{k,m},\cdots,\ddot{\bf A}_{k,m}(Q){\bf x}_{k,m}], \label{Derivation_loc.1}\\
  \frac{\partial{\bf A}_{k,m}{\bf x}_{k,m}}{\partial{\bf y}} & =[\dot{\bf A}_{k,m}(1){\bf x}_{k,m},\cdots,\dot{\bf A}_{k,m}(Q){\bf x}_{k,m}], \label{Derivation_loc.2}
\end{align}
\end{subequations}
where $\ddot{\bf A}_{k,m}(q)=[{\ddot{\bf A}_{\mathrm{R},k,m}}{\bf \Lambda}_{q} {\bf B}{\bf A}_{\mathrm{T},k,m}^T+{\bf A}_{\mathrm{R},k,m}{\bf B}{\bf \Lambda}_{q}{\ddot{\bf A}_{\mathrm{T},k,m}^T}]$ and $\dot{\bf A}_{k,m}(q)=[{\dot{\bf A}_{\mathrm{R},k,m}}{\bf \Lambda}_{q} {\bf B}{\bf A}_{\mathrm{T},k,m}^T+{\bf A}_{\mathrm{R},k,m}{\bf B}{\bf \Lambda}_{q}{\dot{\bf A}_{\mathrm{T},k,m}^T}]$ denotes the partial derivative of ${\bf A}_{k,m}$ w.r.t. $x_q$ and $y_q$, respectively, with each element of the derivative w.r.t $x_q$ being $\frac{\partial{\bf a}_{\mathrm{R},\mathrm{T},\ast}(k,m,q) }{\partial{x_q}}= \ddot{g}_{\mathrm{R},\mathrm{T},\ast}(k,m,q) {\bf a}_{\mathrm{R},\mathrm{T},\ast}(k,m,q)$, where the gain $\ddot{g}_{\mathrm{R},\mathrm{T},\ast}(k,m,q)$ is given by \eqref{derivation_x} arranged at the top of next page
\begin{figure*}[!t]
\begin{subequations} \label{derivation_x}
\begin{flalign} 
  &\dot{g}_{\mathrm{R},\mathrm{T},\ast}(k,m,q) \!=\!(\frac{j2\pi\frac{f_k}{c}(v_{x}^{(q)}mT_{\mathrm{sym}}\!-\!(x_q\!-\!{x}_{\ast}))}{\|{\mathbb I}_{q}-{\mathbb I}_{\ast}\|}-\frac{{x}_q-{x}_{\ast}}{{\|{\mathbb I}_{q}-{\mathbb I}_{\ast}\|^2}} -\frac{j2\pi\frac{f_k}{c}mT_{\mathrm{sym}}(v_{x}^{(q)}({x}_q-{x}_{\ast})^2\!+\!v_{y}^{(q)}({x}_q\!-\!{x}_{\ast})({y}_q\!-\!{y}_{\ast}))}{\|{\mathbb I}_{q}-{\mathbb I}_{\ast}\|^3}). & \\
  &\ddot{g}_{\mathrm{R},\mathrm{T},\ast}(k,m,q) \!=\!(\frac{j2\pi\frac{f_k}{c}(v_{y}^{(q)}mT_{\mathrm{sym}}\!-\!(y_q\!-\!{y}_{\ast}))}{\|{\mathbb I}_{q}-{\mathbb I}_{\ast}\|}-\frac{{y}_q-{y}_{\ast}}{{\|{\mathbb I}_{q}-{\mathbb I}_{\ast}\|^2}} -\frac{j2\pi\frac{f_k}{c}mT_{\mathrm{sym}}(v_{y}^{(q)}({y}_q-{y}_{\ast})^2\!+\!v_{x}^{(q)}({x}_q\!-\!{x}_{\ast})({y}_q\!-\!{y}_{\ast}))}{\|{\mathbb I}_{q}-{\mathbb I}_{\ast}\|^3}). &  
\end{flalign}
\end{subequations}
\hrulefill
\end{figure*}
and each element of the derivation w.r.t $y_q$ can be similarly obtained. 
Substituting \eqref{Derivation_loc.} into \eqref{FIM_spe} and \eqref{FIM_wide-band}, lower bound on CRB of the location estimation for $q$-th target is 
\begin{subequations} \label{CRB_loc.q}
\begin{align}
 \mathrm{CRB}_{x_q}& = \frac{\sigma_r^2}{2\mathbb{E}\{\sum_{k=1}^K\sum_{m=1}^M\|\dot{\bf A}_{k,m}(q){\bf x}_{k,m}\|_2^2\}} ,    \\
 \mathrm{CRB}_{y_q}& = \frac{\sigma_r^2}{2\mathbb{E}\{\sum_{k=1}^K\sum_{m=1}^M\|\ddot{\bf A}_{k,m}(q){\bf x}_{k,m}\|_2^2\}}.
\end{align}
\end{subequations}

With the general CRB expressions given above, we can then explore the CRB for the specific cases. 

\begin{figure*}[!t]
\begin{small}
\begin{subequations} \label{CRB_vel.q_single}
\begin{flalign}
&\mathrm{CRB}_{\alpha_{R_{q}}} \!=\!\mathrm{CRB}_{\alpha_{I_{q}}} =\frac{\sigma_r^2}{2\mathcal{P}{\sum}_{k=1}^K{\sum}_{m=1}^M\|{\bf a}_{\mathrm{R}}(k,m,q)\|_2^2\|{\bf a}_{\mathrm{T}}^{T}(k,m,q)\|_2^2},\\
 &\mathrm{CRB}_{v_{x}}\!=\! \frac{\sigma_r^2}{2|\alpha|^2\{\sum_{k=1}^K\sum_{m=1}^M    (\|\acute{\bf a}_{\mathrm R}(k,m)\|_2^2\|{\bf a}_{\mathrm T}(k,m)\|_2^2\!+\!2\Re{\{\acute{\bf a}_{\mathrm R}^H(k,m){\bf a}_{\mathrm R}(k,m)}{\bf a}_{\mathrm T}^H(k,m)\acute{\bf a}_{\mathrm T}(k,m)\}\!+\!\|{\bf a}_{\mathrm R}(k,m)\|_2^2\|\acute{\bf a}_{\mathrm T}(k,m)\|_2^2) \}},     \label{CRB_vel.q1_single}\\
 &\mathrm{CRB}_{v_{y}} \!=\!\frac{\sigma_r^2}{2|\alpha|^2\{\sum_{k=1}^K\sum_{m=1}^M    (\|\grave{\bf a}_{\mathrm R}(k,m)\|_2^2\|{\bf a}_{\mathrm T}(k,m)\|_2^2\!+\!2\Re{\{\grave{\bf a}_{\mathrm R}^H(k,m){\bf a}_{\mathrm R}(k,m)}{\bf a}_{\mathrm T}^H(k,m)\grave{\bf a}_{\mathrm T}(k,m)\}\!+\!\|{\bf a}_{\mathrm R}(k,m)\|_2^2\|\grave{\bf a}_{\mathrm T}(k,m)\|_2^2) \}}, \label{CRB_vel.q2_single} \\
 &\mathrm{CRB}_{x} \!=\! \frac{\sigma_r^2}{2|\alpha|^2\{\sum_{k=1}^K\sum_{m=1}^M    (\|\dot{\bf a}_{\mathrm R}(k,m)\|_2^2\|{\bf a}_{\mathrm T}(k,m)\|_2^2\!+\!2\Re{\{\dot{\bf a}_{\mathrm R}^H(k,m){\bf a}_{\mathrm R}(k,m)}{\bf a}_{\mathrm T}^H(k,m)\dot{\bf a}_{\mathrm T}(k,m)\}\!+\!\|{\bf a}_{\mathrm R}(k,m)\|_2^2\|\dot{\bf a}_{\mathrm T}(k,m)\|_2^2) \}},    \label{CRB_loc.q1_single}\\
 &\mathrm{CRB}_{y} \!=\!\frac{\sigma_r^2}{2|\alpha|^2\{\sum_{k=1}^K\sum_{m=1}^M    (\|\ddot{\bf a}_{\mathrm R}(k,m)\|_2^2\|{\bf a}_{\mathrm T}(k,m)\|_2^2\!+\!2\Re{\{\ddot{\bf a}_{\mathrm R}^H(k,m){\bf a}_{\mathrm R}(k,m)}{\bf a}_{\mathrm T}^H(k,m)\ddot{\bf a}_{\mathrm T}(k,m)\}\!+\!\|{\bf a}_{\mathrm R}(k,m)\|_2^2\|\ddot{\bf a}_{\mathrm T}(k,m)\|_2^2) \}}. \label{CRB_loc.q2_single}
\end{flalign}
\end{subequations}
\end{small}
\hrulefill
\end{figure*}

\vspace{-0.2cm}
\subsection{Approximated CRB for Single Target Case}
\vspace{-0.1cm}

Note that existing CRB analyses for MIMO radar in the wide-band far-field regime \cite{liang2021CramerRao} and for narrow-band near-field operation \cite{wei2025PartI} typically assume static objects, which limits their relevance to scenarios with non-negligible motion. In Part~II, we therefore investigate the CRB performance of the proposed wide-band near-field framework under a moving-object model, where location–motion coupling must be explicitly accounted for across subcarriers and slow time. Again, we use the abbreviations FF and NF to denote far-field and near-field, respectively. For simplicity, we consider the transmit power for each subcarrier is identical, i.e., $\mathcal{P}_1=\cdots=\mathcal{P}_K=\mathcal{P}$.

\subsubsection{RCS Estimation}
We can first define the gains of Tx and Rx at the $k$-th subcarrier and $m$-th OFDM symbol, respectively, as
\begin{subequations} \label{gain_Tx_Rx}
 \begin{align} 
   G_{\mathrm{Tx}} &\!=\! \|{\bf a}_{\mathrm{T}}(k,m,q)\|_2^2 = \frac{\lambda_k^2}{16\pi^2}\sum_{n_t=1}^{N_t}\frac{1}{r_{n_t}^2} = \frac{\lambda_k^2}{16\pi^2}g_{\mathrm {Tx}}, \label{gain_Tx}\\ 
   G_{\mathrm{Rx}} &\!=\! \|{\bf a}_{\mathrm{R}}(k,m,q)\|_2^2 = \frac{\lambda_k^2}{16\pi^2}\sum_{n_r=1}^{N_r}\frac{1}{r_{n_r}^2} = \frac{\lambda_k^2}{16\pi^2}g_{\mathrm {Rx}}, \label{gain_Rx}
 \end{align}   
\end{subequations}
where $r_{n_t}=\|{\mathbb I}_{q}-{\mathbb I}_{n_t}\|$ and $r_{n_r}=\|{\mathbb I}_{q}-{\mathbb I}_{n_r}\|$. 
Then, we have 
\begin{flalign} \label{CRB_RCS}
\mathrm{CRB}_{\alpha_{R}}=\mathrm{CRB}_{\alpha_{I}}=\frac{256\sigma_r^2\pi^{4}}{2\mathcal{P}Mg_{\mathrm {Tx}}g_{\mathrm {Rx}}\widetilde{\bm\Lambda}}. 
\end{flalign} 
where $\widetilde{\bm\Lambda}_4=\sum_{k=1}^{K}\lambda_k^4$. 
Moreover, we denote $\mathrm{CRB}_{\alpha}=\mathrm{CRB}_{\alpha_{R}}+\mathrm{CRB}_{\alpha_{I}}$ by the CRB for RCS estimation. 

\begin{theorem}
For a single target with co-located Tx/Rx arrays
along with x-axis centered at the origin, the approximated CRBs for RCS
estimation, in terms of far-field and near-field cases, are,
respectively, given by
\begin{subequations} \label{CRB_RCS.2}
\begin{flalign} 
&\mathrm{CRB}_{\alpha}^{\mathrm{FF}}  \approx \frac{256\sigma_r^2\pi^{4}(r_{\mathrm{Tx}}^{\mathrm{FF}}r_{\mathrm{Rx}}^{\mathrm{FF}})^2}{\mathcal{P}M{N_t}{N_r}\widetilde{\bm\Lambda}_4}. \label{CRB_RCS.1.2}\\
&\mathrm{CRB}_{\alpha}^{\mathrm{NF}} \! \approx \! \frac{256\sigma_r^2\pi^{4}(r_{\mathrm{Tx}}^{\mathrm{NF}}r_{\mathrm{Rx}}^{\mathrm{NF}})^2}{\mathcal{P}M{N_t}{N_r}(1+\Delta_{\mathrm{Tx}})(1+\Delta_{\mathrm{Rx}})\widetilde{\bm\Lambda}_4}, \label{CRB_RCS.2.2}
\end{flalign} 
\end{subequations}
where $r_{\mathrm{Tx}}^{\mathrm{FF}}$ and $r_{\mathrm{Rx}}^{\mathrm{FF}}$ denote the range between the far-field target and reference element of Tx and Rx, respectively,  $\Delta_{\mathrm{Tx}}=\frac{(N_t^{2}-1)d_t(4\sin^2\theta_{\mathrm{Tx},q} - 1)}{12 r_{\mathrm{Tx}}^2}$ and $\Delta_{\mathrm{Rx}}=\frac{(N_r^{2}-1)d_r(4\sin^2\theta_{\mathrm{Rx},q} - 1)}{12 r_{\mathrm{Rx}}^2}$.
\end{theorem}

\begin{proof}
The proof is similar to Appendix A in Part I. The only
difference is that the approximation of wide-band steering vector is related to the subcarrier index $k$, which lead to the final wide-band CRB contains $\widetilde{\bm\Lambda}_4$.
So we omit the detailed derivation herein. 
\end{proof}

\subsubsection{Velocity Estimation} Then, for the wide-band velocity CRB, we also have the following Theorem.

\begin{theorem}
Then, for the far-field target, we omit the second-order items $\frac{(N_t^2-1)d_t^2}{12{r_{\mathrm{Tx}}^{\mathrm{FF}}}^2}$ and $\frac{(N_r^2-1)d_r^2}{12{r_{\mathrm{Rx}}^{\mathrm{FF}}}^2}$, and have
\begin{small}
\begin{subequations}
\label{CRB_vel_FF_final}
\begin{align}
\hspace*{-0.5em}\mathrm{CRB}_{v_x}^{\mathrm{FF}}
&\!\approx\!
\frac{32\pi^2\sigma_r^2(r_{\mathrm{Tx}}^{\mathrm{FF}}r_{\mathrm{Rx}}^{\mathrm{FF}})^2}
{|\alpha|^2\mathcal{P}N_tN_rT_{\mathrm{sym}}^2
C_M
\big(\sin\theta_{\mathrm{Tx},q}\!+\!\sin\theta_{\mathrm{Rx},q}\big)^2\widetilde{\bm\Lambda}_2}, \\
\hspace*{-0.5em}\mathrm{CRB}_{v_y}^{\mathrm{FF}}
&\!\approx\!
\frac{32\pi^2\sigma_r^2(r_{\mathrm{Tx}}^{\mathrm{FF}}r_{\mathrm{Rx}}^{\mathrm{FF}})^2}
{|\alpha|^2\mathcal{P}N_tN_rT_{\mathrm{sym}}^2
C_M
\big(\cos\theta_{\mathrm{Tx},q}\!+\!\cos\theta_{\mathrm{Rx},q}\big)^2\widetilde{\bm\Lambda}_2}.
\end{align}
\end{subequations}
\end{small}
\noindent\hspace{-0.5em}where $C_M=\frac{M(M+1)(2M+1)}{6}$ and $\widetilde{\bm\Lambda}_2=\sum_{k=1}^{K}\lambda_k^2$.
Then, for the near-field target, we have
\begin{small}
\begin{subequations}
\label{CRB_vel_NF_final}
\begin{align}
&\hspace*{-1em}\mathrm{CRB}_{v_x}^{\mathrm{NF}}
\approx
\frac{32\pi^2\sigma_r^2({r{\mathrm{Tx}}^{\mathrm{NF}}}{r_{\mathrm{Rx}}^{\mathrm{NF}}})^2}
{|\alpha|^2\mathcal{P}N_tN_rT_{\mathrm{sym}}^2
C_M\big(\sin\theta_{\mathrm{Tx},q}\!+\!\sin\theta_{\mathrm{Rx},q}\big)^2
\Psi\widetilde{\bm\Lambda}_2},\\
&\hspace*{-1em}\mathrm{CRB}_{v_y}^{\mathrm{NF}} 
\approx
\frac{32\pi^2\sigma_r^2({r{\mathrm{Tx}}^{\mathrm{NF}}}{r_{\mathrm{Rx}}^{\mathrm{NF}}})^2}
{|\alpha|^2\mathcal{P}N_tN_rT_{\mathrm{sym}}^2
C_M\big(\cos\theta_{\mathrm{Tx},q}\!+\!\cos\theta_{\mathrm{Rx},q}\big)^2
\Psi\widetilde{\bm\Lambda}_2},
\end{align}
\end{subequations}
\end{small}
where the near-field correction terms $\Psi_{x}$ and $\Psi_{y}$ are same as the definition in Part I, see \cite{wei2025PartI} in detail.
\end{theorem}

\begin{proof}
    See Appendix B in Part I \cite{wei2025PartI}.
\end{proof}

\subsubsection{Location Estimation}
To obtain tractable closed-form CRBs, we retain only the dominant phase derivative with respect to the path length and neglect both the amplitude derivative and the higher-order Doppler–location coupling terms. Similarly, for the wide-band CRB of location estimation, we have the following Theorem. 

\begin{theorem}
We have the far-field wide-band CRB approximation for localization as  
\begin{small}
\begin{subequations}
\label{CRB_loc_FF_NB}
\begin{align}
\mathrm{CRB}_{x}^{\mathrm{FF}}
&\approx
\frac{32\pi^2\sigma_r^2
\big(r_{\mathrm{Tx}}^{\mathrm{FF}}r_{\mathrm{Rx}}^{\mathrm{FF}}\big)^2}
{|\alpha|^2\mathcal PMN_tN_r
\Big(\sin\theta_{\mathrm{Tx},q}+\sin\theta_{\mathrm{Rx},q}\Big)^2
\widetilde{\bm\Lambda}_2},\\
\mathrm{CRB}_{y}^{\mathrm{FF}}
&\approx
\frac{32\pi^2\sigma_r^2
\big(r_{\mathrm{Tx}}^{\mathrm{FF}}r_{\mathrm{Rx}}^{\mathrm{FF}}\big)^2}
{|\alpha|^2\mathcal PMN_tN_r
\Big(\cos\theta_{\mathrm{Tx},q}+\cos\theta_{\mathrm{Rx},q}\Big)^2
\widetilde{\bm\Lambda}_2}.
\end{align}
\end{subequations}
\end{small}
Similarly, for the near-field target, we have 
\begin{small}
\begin{subequations}
\label{CRB_loc_NF_NB}
\begin{align}
\mathrm{CRB}_{x}^{\mathrm{NF}}
&\approx
\frac{32\pi^2\sigma_r^2
\big(r_{\mathrm{Tx}}^{\mathrm{NF}}r_{\mathrm{Rx}}^{\mathrm{NF}}\big)^2}
{|\alpha|^2\mathcal PMN_tN_r
\Big(\sin\theta_{\mathrm{Tx},q}+\sin\theta_{\mathrm{Rx},q}\Big)^2
\Psi_x\widetilde{\bm\Lambda}_2},\\
\mathrm{CRB}_{y}^{\mathrm{NF}}
&\approx
\frac{32\pi^2\sigma_r^2
\big(r_{\mathrm{Tx}}^{\mathrm{NF}}r_{\mathrm{Rx}}^{\mathrm{NF}}\big)^2}
{|\alpha|^2\mathcal PMN_tN_r
\Big(\cos\theta_{\mathrm{Tx},q}+\cos\theta_{\mathrm{Rx},q}\Big)^2
\Psi_y\widetilde{\bm\Lambda}_2}.
\end{align}
\end{subequations}
\end{small} 
\end{theorem}
\begin{proof}
    We omit the proof due to the similarity to Appendix B in Part I.
\end{proof}


\begin{remark}
In wide-band systems, the Fisher information aggregates across orthogonal subcarriers, which introduces wavelength-sum factors such as
$\widetilde{\Lambda}_2 \triangleq \sum_{k=1}^{K}\lambda_k^{2}$ and
$\widetilde{\Lambda}_4 \triangleq \sum_{k=1}^{K}\lambda_k^{4}$ in the resulting CRB expressions. 
For any nonzero bandwidth, these sums satisfy
$\widetilde{\Lambda}_2 > K\lambda_c^{2}$ and $\widetilde{\Lambda}_4 > K\lambda_c^{4}$ (with equality only in the degenerate single-tone case), implying that wide-band signaling provides strictly larger information than the narrow-band case under the same CPI and average transmit power. 
Consequently, the wide-band CRBs for RCS, velocity, and localization derived in Part~II are generally lower than their narrow-band counterparts reported in Part~I; see \cite{wei2025PartI}.
\end{remark}


\section{Numerical Experiments}
 
Throughout the simulations, the system parameters are configured according to the potential 6G frequency bands identified by the ITU-R and recent studies on upper-mid (FR3) spectrum allocation \cite{testolina2024sharing,emil2025enabling}.

\begin{table}[!t] \label{tab:param}
\caption{Simulation Parameters Setting}
\begin{center}
  \begin{tabular}{p{5.0cm}|p{1.0cm}|p{1.5cm}}\hline\hline
  Parameter                          & Symbol               &  Value \\   \hline
  Central frequency                  & $f_c$                &  15 GHz \\   \hline
  Subcarrier spacing                 & $\triangle{f}$       &  120 KHz \\   \hline 
  Element No. of Tx                  & $N_{\mathrm{Tx}}$    &  256   \\   \hline  
  Element No. of Rx                  & $N_{\mathrm{Rx}}$    &  256        \\   \hline
  No. of sensing targets             & $Q$                  &  1$\thicksim$3   \\   \hline 
  No. of subcarriers                 & $K$                  &  128   \\   \hline 
  No. of OFDM symbol                 & $M$                  &  254   \\   \hline   
  Noise power                        & $\sigma_r^2$         &  -114 dBm \\     \hline  
  Transmit power at each subcarrier  & $\mathcal{P}_k$      &  20 dBm \\   \hline  
  \end{tabular}
\end{center}
\end{table}

\vspace{-0.2cm}
\subsection{Wide-band Parameter Settings}
\vspace{-0.1cm} 
The carrier frequency is set to $f_c=15$~GHz and the speed of light is $c=3\times 10^8$~m/s. We employ an OFDM waveform with subcarrier spacing $\Delta f=120$~kHz and $K=128$ subcarriers, yielding an occupied bandwidth $B=K\Delta f=15.36$~MHz. Each coherent processing interval (CPI) consists of $M=256$ OFDM symbols with symbol duration $T_{\mathrm{sym}}=1/\Delta f$, i.e., $T_{\mathrm{CPI}}=MT_{\mathrm{sym}}$. The subcarrier frequency grid is centered around $f_c$ as $f_k=f_c+\left(k-\frac{K-1}{2}\right)\Delta f$, and the corresponding wavelengths are $\lambda_k=c/f_k$.
Both Tx and Rx employ centered ULAs along the $x$-axis with $N_t=N_r=256$ elements and half-wavelength spacing $d=\lambda_c/2$, where $\lambda_c=c/f_c$. The Tx/Rx element coordinates are ${\mathbb I}_{n_t}=(x_{n_t},0)$ and ${\mathbb I}_{n_r}=(x_{n_r},0)$ with $x_{n}=\left(n-\frac{N-1}{2}\right)d$. The array aperture is $D=(N-1)d$, and the reactive and radiative near-field boundaries are computed via $0.62\sqrt{D^3/\lambda_c}$ and $2D^2/\lambda_c$, respectively. Additive noise is modeled as spatially and temporally white circular Gaussian with variance ${\sigma}_r^2 = -114$ dBm and transmit power is fixed to $\mathcal{P}=0.1$ W. The relative error is calculated by $\epsilon_{\mathrm{rel}} =\left|\left(\mathrm{CRB}_{\mathrm{approx.}}-\mathrm{CRB}_{\mathrm{true}}\right)/\mathrm{CRB}_{\mathrm{true}}\right|$.

\vspace{-0.2cm}
\subsection{CRB with Single Target and Monostatic}
\vspace{-0.1cm}
First, let us consider the single target and monostatic setup, i.e., Tx and Rx have the same antenna and are located in the same position. 
For the single target case, we set the RCS, velocity and location as ${\alpha}=1+0.1j$,  ${\bf v}=[1,4]$ m/s and $(20^\circ,100~\mathrm{m})$, respectively.  Hereafter, we use the terms "True", "FF" and "NF" denote the true CRB, far-field approximation, and near-field approximation, respectively.   

Fig.~\ref{RCS_range} illustrates the CRB for estimating the target RCS versus range in the wide-band near-field setting, where multiple bandwidth configurations  ($K\in\{8,32,64,128\}$) are considered by varying the number of subcarriers  while keeping the subcarrier spacing fixed. The true CRB increases monotonically with range for all $K$, which is mainly due to the reduced received signal strength as the propagation distance grows. Importantly, increasing $K$ consistently reduces the bound, demonstrating the benefit of frequency diversity in wide-band sensing: a larger occupied bandwidth provides additional phase and range sensitivity across subcarriers and increases the total Fisher information accumulated over frequency.
Fig.~\ref{RCS_range} also compares the proposed near-field (NF) approximation and the far-field (FF) approximation with the true CRB and reports their relative errors on the right axis. The NF approximation remains closely aligned with the true CRB across the full range interval for all $K$, yielding uniformly small errors. In contrast, the FF approximation exhibits a more pronounced mismatch in the near-field region (short-to-moderate ranges), while its accuracy improves at larger ranges as the propagation approaches the plane-wave regime. Overall, Fig.~1 confirms that wide-band operation significantly improves RCS estimation, and that near-field-consistent modeling (or the proposed NF approximation) is necessary to accurately characterize the fundamental limits in the radiative near-field, especially when the bandwidth and aperture are large.

\begin{figure}[t]
\centering{\includegraphics[width=1\columnwidth]{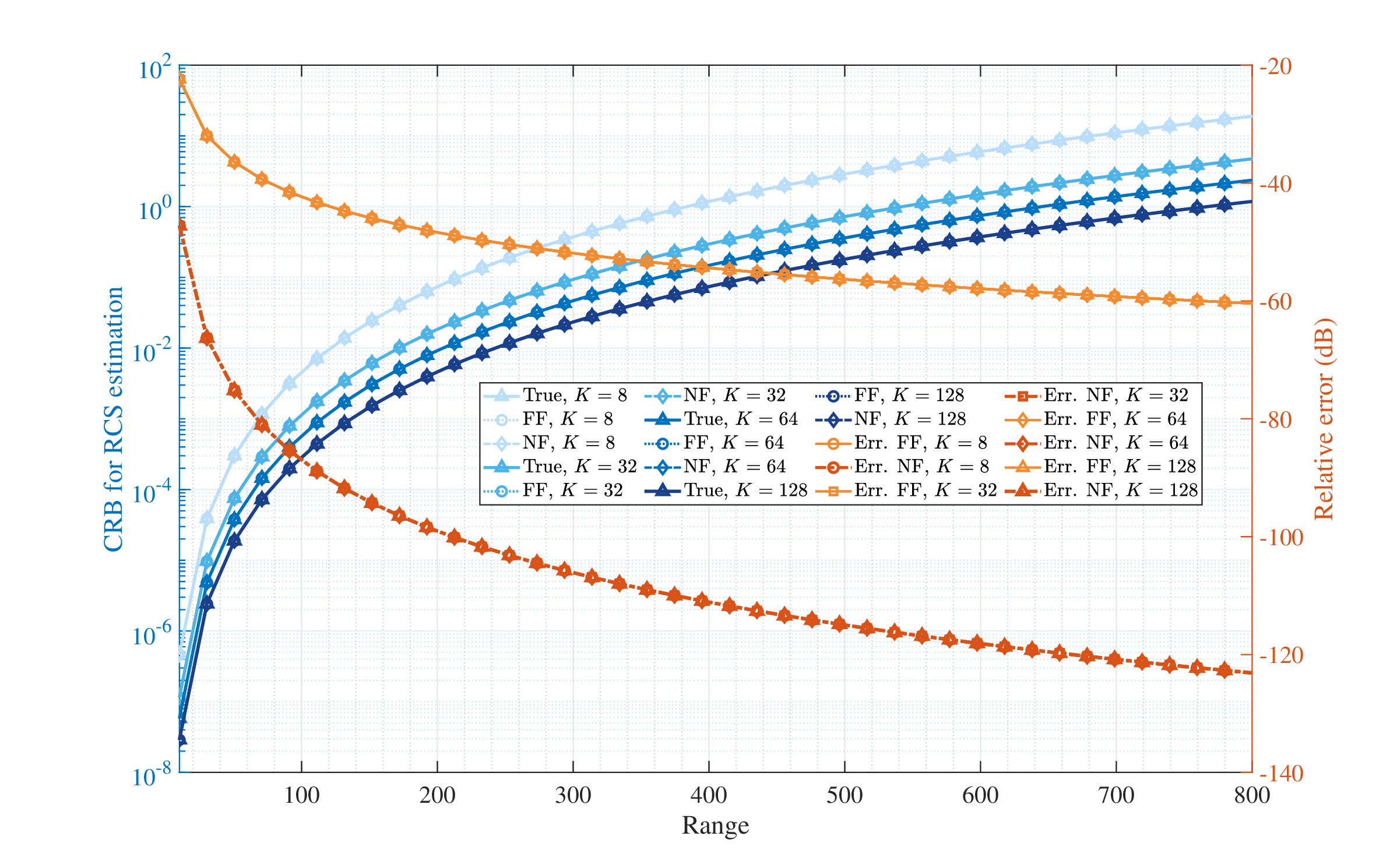}}
\caption{CRB for RCS versus range. 
\label{RCS_range}}
\end{figure}

Fig.~\ref{RCS_antenna} shows the CRB for estimating the target reflectivity (RCS) versus the number of antennas for different bandwidth configurations. The true CRB decreases monotonically with the antenna number for all $K$, confirming that enlarging the aperture improves RCS estimation through higher array gain and increased spatial diversity. Moreover, for any fixed antenna size, increasing $K$ further tightens the bound, highlighting the complementary benefits of spatial and frequency diversity in wide-band near-field sensing.
The right axis reports the relative error of the proposed NF approximation. The NF approximation closely matches the true CRB across the full antenna-number range for all $K$, with consistently small errors. This indicates that the approximation remains reliable even as the aperture grows electrically large and the model must capture pronounced near-field curvature effects. Overall, Fig.~\ref{RCS_antenna} verifies that wide-band operation and large apertures jointly enhance RCS estimation accuracy, and that the proposed NF approximation provides an accurate and efficient surrogate for performance evaluation and design.

\begin{figure}[t]
\centering{\includegraphics[width=1\columnwidth]{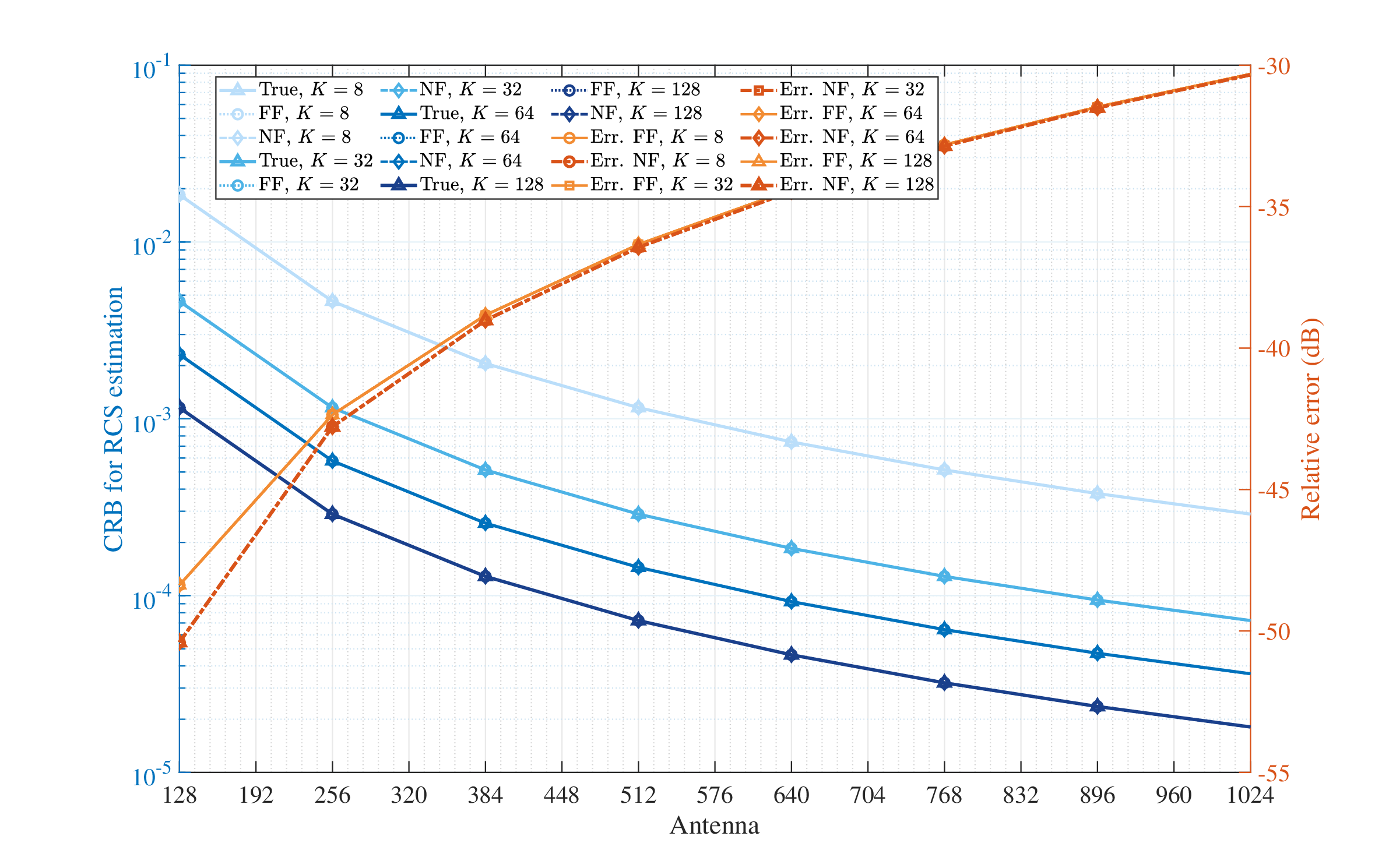}}
\caption{CRB for RCS versus antenna number. 
\label{RCS_antenna}}
\end{figure}

Fig.~\ref{vel_range} depicts the CRBs for estimating the velocity components versus target range in the wide-band near-field regime, for different numbers of subcarriers. For all $K$, the true velocity CRBs increase monotonically with range, indicating that velocity estimation becomes progressively more difficult as the target moves farther away. This trend is driven by the reduced received signal strength and the weakened sensitivity of the slow-time phase evolution to motion at larger distances. Increasing the bandwidth (larger $K$) consistently improves the velocity bounds across the entire range interval. This shows that frequency diversity does not only benefit ranging-related parameters, but also enhances motion estimation by providing additional phase variation across subcarriers and strengthening the overall Fisher information when the wide-band near-field model accounts for frequency-dependent phase and amplitude.

The figure also reports the relative errors of the NF and FF approximations. The proposed NF approximation remains tightly aligned with the true CRB for all $K$, yielding uniformly small errors, whereas the FF approximation exhibits noticeably larger mismatch at short-to-moderate ranges and gradually improves as the target transitions toward the far-field. Overall, Fig.~~\ref{vel_range} confirms the joint benefits of wide-band operation and near-field-consistent modeling for characterizing velocity estimation limits in large-aperture sensing systems.

\begin{figure}[t]
\centering{\includegraphics[width=1\columnwidth]{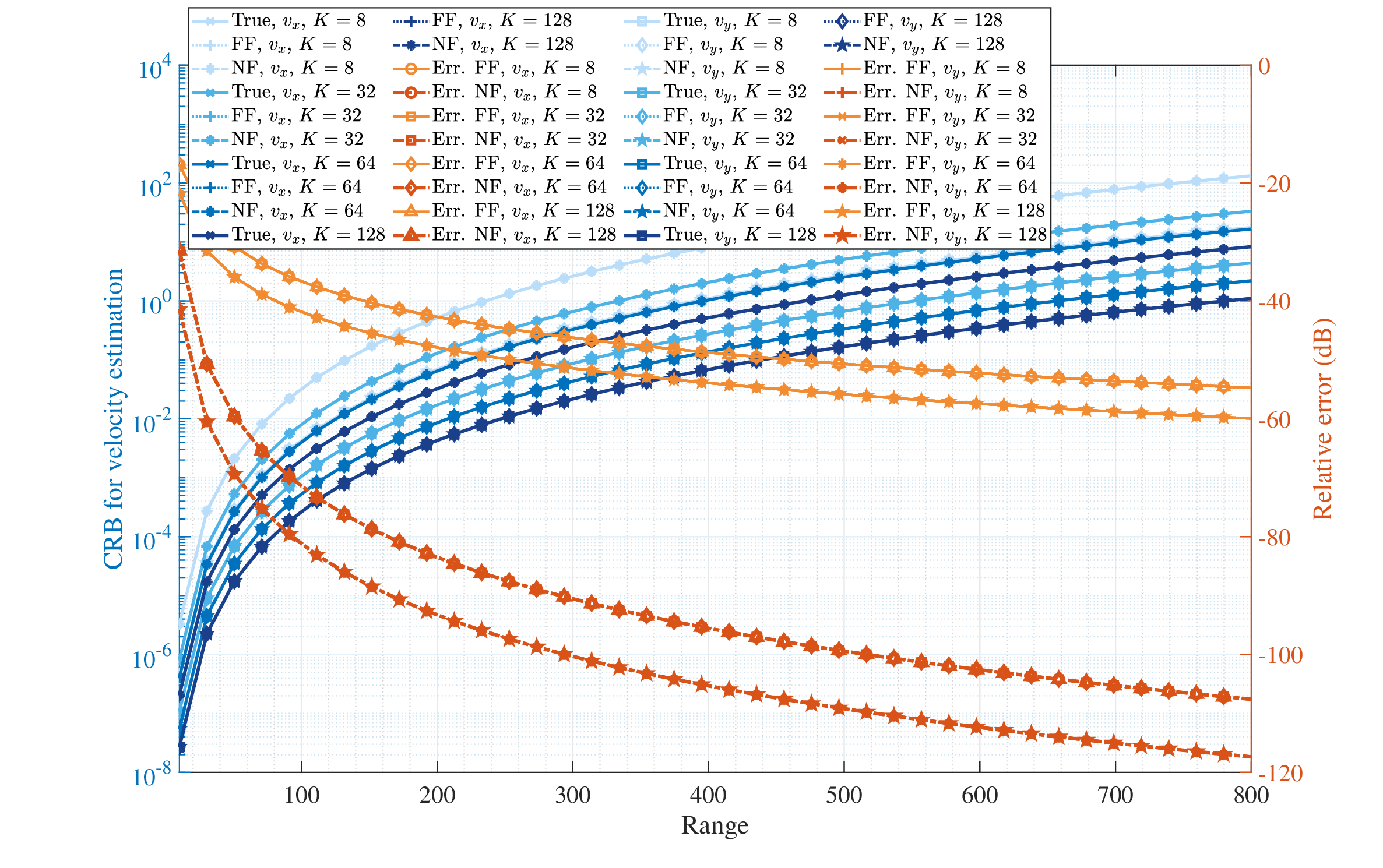}}
\caption{CRB for velocity versus range. 
\label{vel_range}}
\end{figure}

Fig.~\ref{vel_antenna} reports the velocity CRBs versus the number of antennas for different bandwidth configurations. The true CRBs for both velocity components decrease monotonically with the antenna number, confirming that enlarging the aperture improves velocity estimation through higher array gain and increased spatial sensitivity to the Doppler-induced phase evolution. For any fixed antenna size, increasing $K$ further tightens the bounds, illustrating the complementary roles of spatial and frequency diversity in wide-band near-field sensing.

The right axis shows the relative errors of the approximations. The proposed NF approximation remains tightly matched to the true CRB across the full antenna-number range for all $K$, with consistently small errors. In contrast, the FF approximation becomes less reliable as the array size grows and the aperture becomes electrically large, since the plane-wave model cannot capture the near-field curvature and the resulting parameter coupling. Overall, Fig.~~\ref{vel_antenna} demonstrates that accurate near-field modeling is essential for characterizing velocity limits in large-aperture wide-band systems, while the NF approximation provides an efficient and accurate surrogate for the practical design studies.

\begin{figure}[t]
\centering{\includegraphics[width=1\columnwidth]{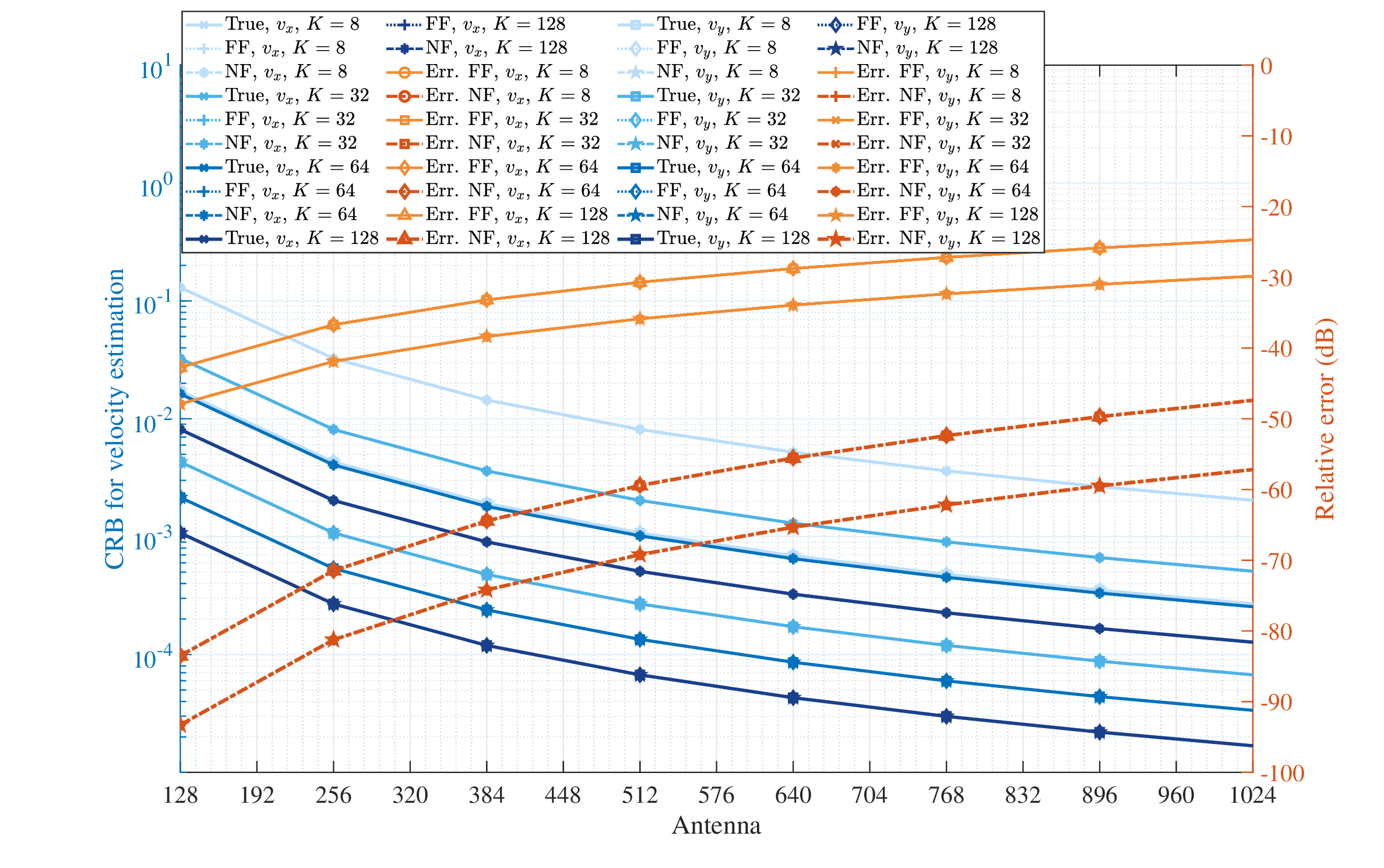}}
\caption{CRB for velocity versus range. 
\label{vel_antenna}}
\end{figure}

\begin{figure}[t]
\centering{\includegraphics[width=1\columnwidth]{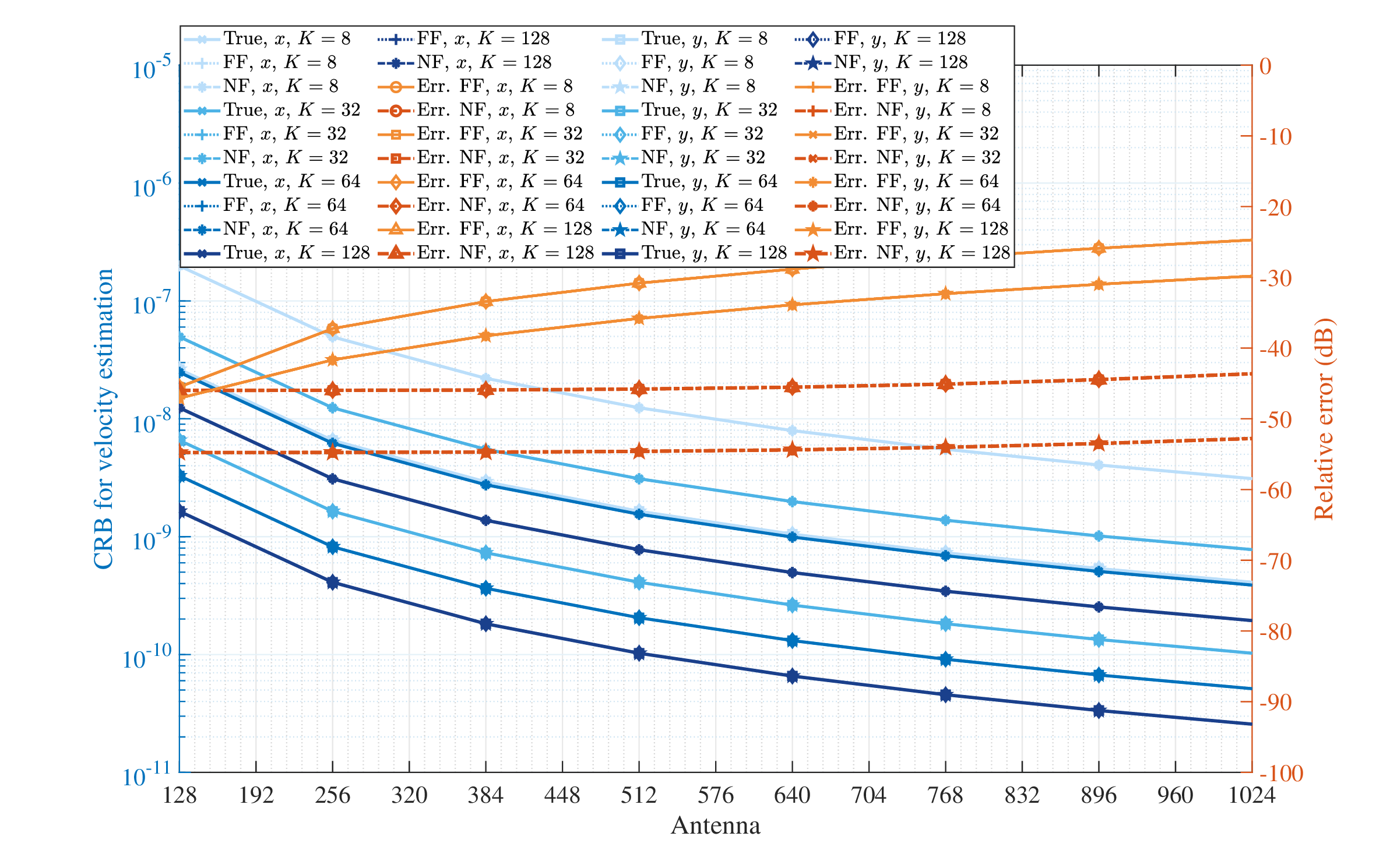}}
\caption{CRB for location versus antenna. 
\label{loc_antenna}}
\end{figure}

Fig.~\ref{loc_antenna} presents the CRBs for estimating the target location components versus the number of antennas under different bandwidth configurations. The true location CRBs decrease monotonically as the antenna number increases, confirming that a larger aperture improves localization accuracy by providing higher array gain and stronger spatial phase gradients. In addition, for any fixed antenna size, increasing $K$ yields consistently tighter bounds, demonstrating that frequency diversity in wide-band sensing complements spatial diversity and enhances the overall sensitivity to position through frequency-dependent phase variations.

The right axis reports the relative errors of the approximations. The proposed NF approximation remains closely matched to the true CRB across the entire antenna-number range for all $K$, indicating that it reliably captures the near-field curvature effects even for electrically large arrays. In contrast, the FF approximation exhibits noticeably larger deviations and tends to become less accurate as the array grows, since the plane-wave model neglects the element-dependent distance variations that dominate near-field localization. Overall, Fig.~\ref{loc_antenna} highlights the joint benefits of large apertures and wide bandwidth for near-field localization, and validates the accuracy of the proposed NF approximation for efficient performance evaluation.

\vspace{-0.2cm}
\subsection{CRB with Multiple Targets and Bistatic}
\vspace{-0.1cm}
Now, we consider bistatic setup with three targets, where the Tx and Rx centroids are separated (Tx at $x_{t_0}=-2$, Rx at $x_{r_t}=2$). Three targets are located at $(20^\circ,100~\mathrm{m})$ with ${\bm v}=[1,4]$~m/s, $(-45^\circ,150~\mathrm{m})$ with ${\bm v}=[4,3]$~m/s, and $(-5^\circ,50~\mathrm{m})$ with ${\bm v}=[10,6]$~m/s. We calculate the total CRB and approximation error by summing the CRB and error for all targets.

Fig.~\ref{RCS_antenna2} compares the CRB for RCS estimation versus the number of antennas under the true model and the FF approximation (with corresponding relative error on the right axis). The true CRB decreases monotonically with the antenna number, confirming that enlarging the aperture improves reflectivity estimation through higher array gain and increased measurement diversity. Both FF approximation and NF approximation reach a good approximation in bistatic and multi-target case. 


\begin{figure}[t]
\centering{\includegraphics[width=1\columnwidth]{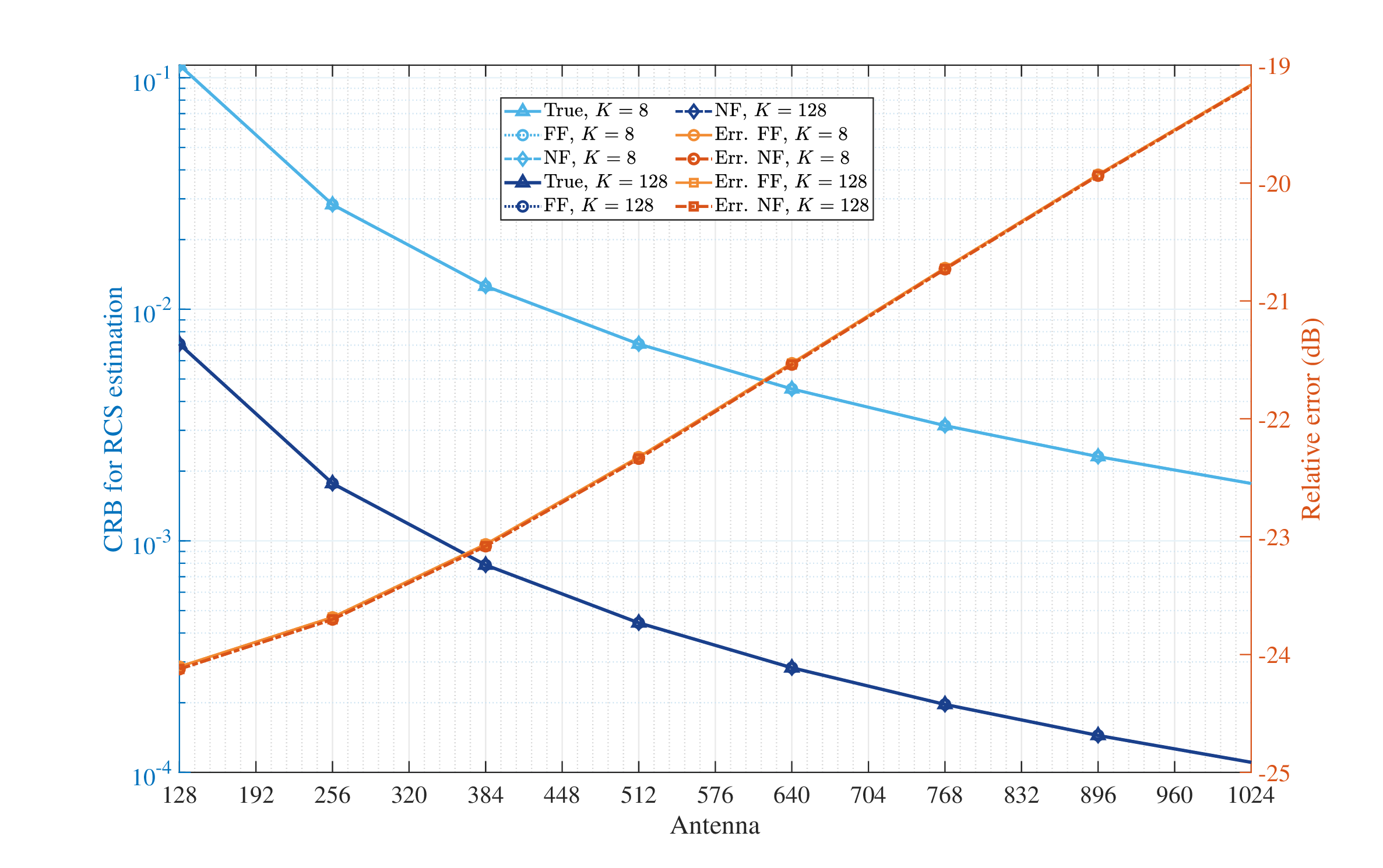}}
\caption{CRB for RCS versus antenna. 
\label{RCS_antenna2}}
\end{figure}

Fig.~\ref{vel_antenna2} shows the velocity CRBs versus the number of antennas, comparing the true model with the FF approximation and reporting the corresponding relative error on the right axis. The true CRBs decrease monotonically with antenna number, indicating that larger apertures provide higher array gain and stronger spatial sensitivity to the Doppler-induced phase evolution, thereby improving velocity estimation for bistatic and multi-target case.

The FF approximation exhibits a clear mismatch whose relative error increases as the array grows. This behavior reflects the breakdown of the plane-wave assumption for electrically large apertures: the element-dependent near-field geometry induces non-negligible curvature and local line-of-sight variations that affect the Doppler phase across the array and couple motion with position. As a result, FF-based bounds can substantially misestimate the achievable velocity accuracy for large arrays, motivating near-field-consistent CRB analysis in the wide-band regime.

\begin{figure}[t]
\centering{\includegraphics[width=1\columnwidth]{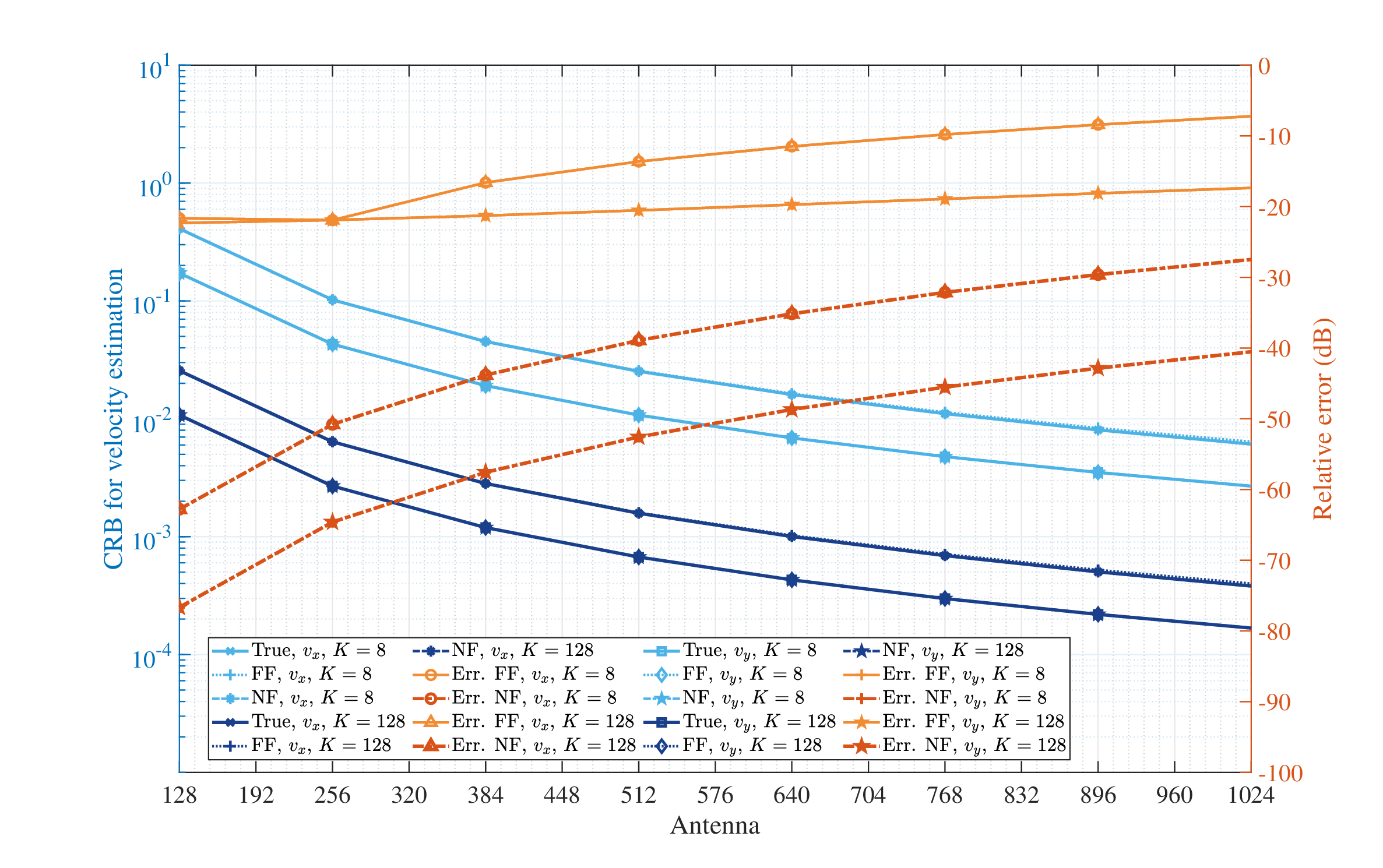}}
\caption{CRB for velocity versus antenna. 
\label{vel_antenna2}}
\end{figure}

\begin{figure}[t]
\centering{\includegraphics[width=1\columnwidth]{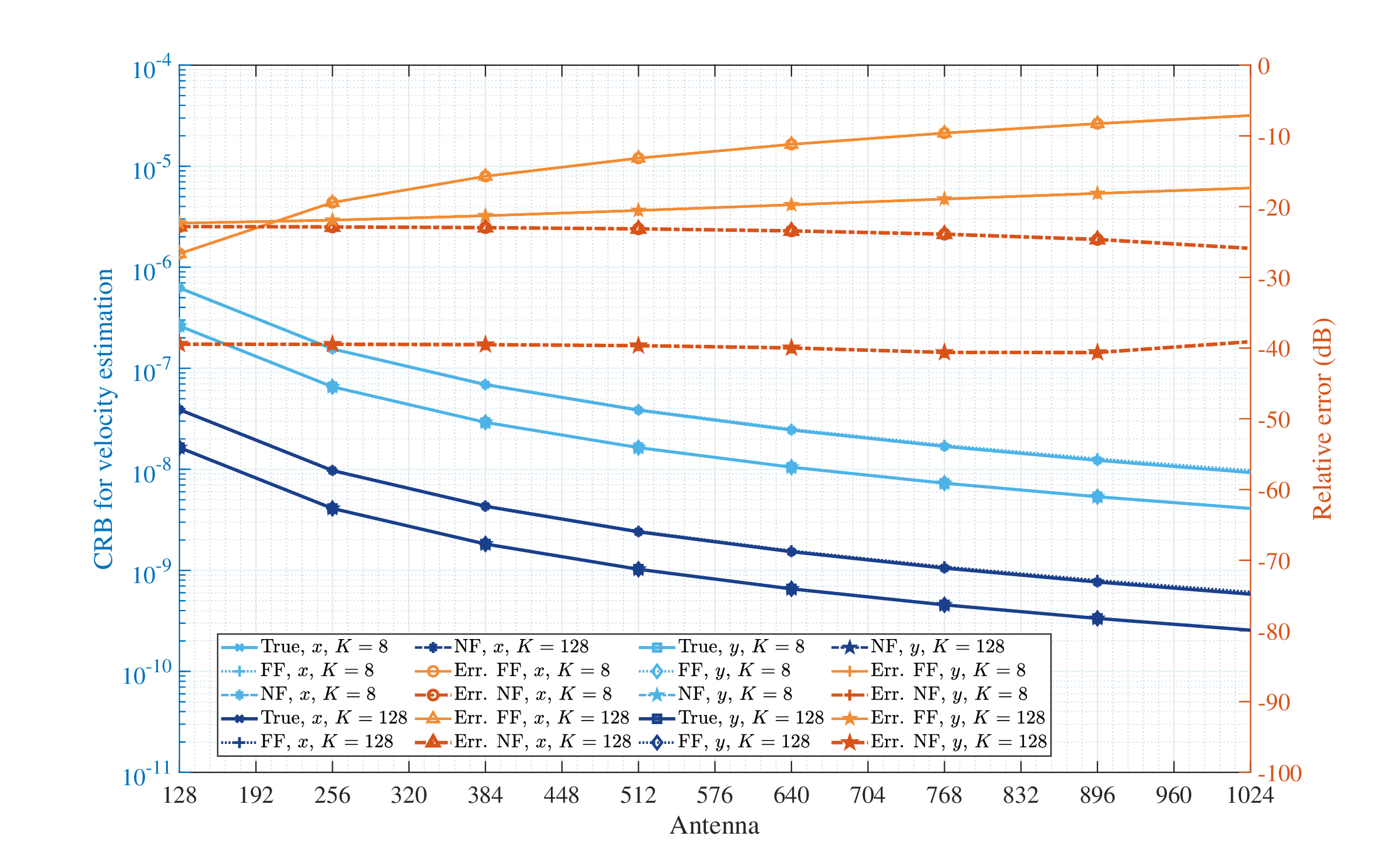}}
\caption{CRB for location versus antenna. 
\label{loc_antenna2}}
\end{figure}

Fig.~\ref{loc_antenna2} reports the location CRBs versus the number of antennas, comparing the true model with the FF approximation and showing the corresponding relative error on the right axis. The true CRBs decrease monotonically as the antenna number increases, confirming that enlarging the aperture improves localization accuracy by strengthening the spatial phase gradients and increasing the effective measurement diversity.

In contrast, the FF approximation shows a noticeable mismatch whose relative error increases with antenna number. As the aperture becomes electrically large, near-field curvature and element-dependent distances dominate the position sensitivity, while the plane-wave FF model suppresses these effects and yields overly simplified Fisher information. Consequently, FF-based localization bounds can become increasingly unreliable for large arrays, especially for $x$-axis, highlighting the need for near-field-consistent CRB characterizations (or accurate NF approximations) when evaluating wide-band near-field localization performance.

\section{Conclusions}
This paper (Part II) studied fundamental limits for wide-band near-field sensing under an OFDM-based model with potentially separated transmit and receive arrays. By preserving spherical-wave propagation and an element-wise Doppler model, we derived the wide-band FIM and obtained conditional CRBs for joint estimation of location, 2D velocity, and complex reflectivity of multiple moving targets. The bounds explicitly show how sensing information accumulates over CPI and across subcarriers, and how near-field geometry induces intrinsic position–velocity coupling. We further developed far-field and near-field asymptotic CRB approximations that expose the dominant scaling laws with respect to array aperture, observation time, carrier wavelength, and bandwidth, providing design-oriented insights and enabling accurate regime comparisons. Numerical results validated the analysis and quantified the approximation accuracy.

Part I established the narrow-band baseline and its far- and near-field scaling laws using a snapshot (slow-time) model. Part II extends this framework to wide-band operation by incorporating frequency-dependent propagation and showing that the overall information aggregates across subcarriers under OFDM orthogonality, while also highlighting when residual delay spread may introduce ISI and affect the validity of the per-subcarrier model. 

Several extensions are worth pursuing. First, it is of interest to incorporate colored or noncircular noise and residual self-interference models that yield nontrivial covariance terms in the FIM. Second, extending the analysis to multipath-rich environments and higher-order reflections would enable bounds under more realistic propagation. Third, wide-band near-field operation with strong dispersion motivates CRBs under ISI-aware observation models beyond the ideal orthogonal-subcarrier assumption. Finally, integrating these bounds into waveform/beamforming and synthetic-aperture trajectory design provides a principled route to near-field sensing strategies that approach the derived fundamental limits.

\ifCLASSOPTIONcaptionsoff
  \newpage
\fi

\bibliographystyle{IEEEtran}
\bibliography{ref}

\ifCLASSOPTIONcaptionsoff
  \newpage
\fi

\end{document}